\pdfoutput=1

\documentclass{acmsmallarxiv}
\usepackage{ifthen}

\usepackage{multirow}
\usepackage{pdfsync} 
\usepackage{mdwlist}
\usepackage{amsmath}
\usepackage{arydshln}
\usepackage{fix-cm}
\usepackage{rotating}
\usepackage{subfloat}
\usepackage{subfig} 
\usepackage{algorithm}

\usepackage{verbatim}
\usepackage{moreverb}

\usepackage{algorithmic} % can use noend to play with 
 \usepackage{url}
 \usepackage{graphicx}
\graphicspath{{./ScalingMultipleLists/}{scripts/}{plots/}{.}{../BitmapIndexCpp/dbash/}{../BitmapIndexCpp/bash/}{../BitmapIndexCpp/dbash/RLEmodelling/}{ExperimentalSoftware/}{../JavaRowReordering/ExperimentalResults/}{../jrr/ExperimentalResults/}{../JavaRowReordering/ExperimentalResults/LongRuns/}{../jrr/ExperimentalResults/LongRuns/}{../EdStuff/figures/}}
\usepackage{amssymb}
\usepackage{ifpdf}
\begin{document}
\def\mytitle{Reordering Rows for Better Compression: Beyond the Lexicographic Order}
\markboth{D. Lemire et al.}{\mytitle}
\title{\mytitle}
\author{DANIEL LEMIRE
\affil{TELUQ}
OWEN KASER
\affil{University of New Brunswick}
EDUARDO GUTARRA
\affil{University of New Brunswick}}

\begin{abstract}
Sorting database tables before compressing them improves
the compression rate. Can we do better than the lexicographical
order?
For minimizing the number of runs in a run-length encoding compression scheme, 
the best approaches to row-ordering are derived from 
traveling salesman heuristics, although there is a significant trade-off between running time and compression.
A new heuristic, \textsc{Multiple Lists}, which is a variant on \textsc{Nearest Neighbor} that trades off  compression for a major running-time speedup, is a good option for very large tables. 
However, for some compression schemes, it is more important to generate long
runs rather than few runs. For this case, another novel heuristic, \textsc{Vortex},
is promising.  
We find that we can improve run-length encoding up to a factor of 3 whereas we can improve prefix coding by up to 80\%: these gains are on top of the gains due to 
 lexicographically sorting the  table. We  prove that the new row reordering is optimal (within 10\%) at minimizing the runs of identical values within columns,  in a few cases.
\end{abstract}


%sorted alphabetically
 
%note case on title, cannot use the \mytitle
\acmformat{Lemire, D., Kaser, O., Gutarra, E. 2012. Reordering rows for better compression: beyond the lexicographic order.} %per acmsmall

\maketitle 
   
\begin{bottomstuff}
This work is supported by Natural Sciences and Engineering Research Council of Canada %National Research Council of Canada, 
grants 261437 and 155967 and a Quebec/NB Cooperation grant.

Author's addresses: D. Lemire, LICEF Research Center, TELUQ; 
O. Kaser  {and} E. Gutarra, Dept.\ of CSAS, University of New Brunswick, Saint John.
\end{bottomstuff}

%\maketitle %here per acmsmall and the example, but this is weird.

 \section{Introduction}

Database compression reduces storage while  improving the performance of
some queries. 
It is commonly recommended to sort tables
to improve
the compressibility of the tables~\cite{Poess:2003:DCO:1315451.1315531} or of the indexes~\cite{arxiv:0901.3751}.
While it is not always possible or useful to sort and compress tables, 
sorting is a critical component of some column-oriented architectures~\cite{abadi2008column,1453912}

At the simplest level, we model compressibility by counting
runs of identical values within columns. 
Thus, we want to reorder rows to minimize the total number
of runs, in all columns (\S~\ref{sec:minimizing-runs}). 
The lexicographical order is the most common row-reordering heuristic for this problem.

Can we beat the lexicographic order? Engineers might be willing to spend extra time reordering rows---even for
modest gains (10\% or 20\%)---if the computational cost is acceptable. Indeed, popular compression utilities such as bzip2 are often several times slower than faster alternatives (e.g., gzip) for similar gains.

Moreover, minimizing the number of runs is of theoretical interest.
Indeed, it  reduces to the Traveling
Salesman Problem (TSP) under the Hamming distance---an NP-hard 
problem~\cite{258541,ernvall1985np} (\S~\ref{sec:runminandtsp}).
Yet there have been few attempts to design and study TSP heuristics with the Hamming 
distance 
and even fewer
 on large data sets. 
For the generic TSP, there are several well known heuristics (\S~\ref{sec:tsp-heuristics}),
as well as strategies to scale them up (\S~\ref{sec:horizontalpartition}).
Inspired by these heuristics, we introduce the novel \textsc{Multiple Lists} heuristic (\S~\ref{sec:approx-nn}) which is designed with the Hamming distance and scalability in mind.

While counting runs is convenient, 
it is an incomplete model.
Indeed, several compression algorithms for databases may be more effective when there are many ``long runs''~(\S~\ref{sec:blockwiseruns}).
Thus, instead of minimizing the number of runs of column values, we may seek to maximize
the number of long runs. We can then test the result with popular compression algorithms~(\S~\ref{sec:realistic-column-storage}). For this new problem, we propose two heuristics: 
   \textsc{Frequent-Component} (\S~\ref{sec:fc}), and \textsc{Vortex} (\S~\ref{sec:vortex}).
   \textsc{Vortex} is novel.

All our contributed heuristics have $O(n \log n)$ complexity when the number of columns
is a constant. However, \textsc{Multiple Lists} uses   
a linear number of random accesses,  making it prohibitive for very large tables: in such cases, 
we use table partitioning~(\S~\ref{sec:Partitioning}).

We can assess these TSP heuristics experimentally under the Hamming distance~(\S~\ref{sec:experiments}).
Using synthetic data sets (uniform and Zipfian), we find that   \textsc{Vortex} is a competitive
heuristic on  Zipfian data. It is one of the best
heuristics for generating long runs. 
Meanwhile, \textsc{Multiple Lists} offers a good compromise between speed and run minimization: it can
even surpass much more expensive alternatives. Unfortunately, it is poor at generating long runs.

Based on these good results, we apply \textsc{Vortex} and \textsc{Multiple Lists} to 
realistic tables, using various table-encoding techniques.
We show that on several data sets both \textsc{Multiple Lists} and
 \textsc{Vortex} can improve compression  when compared to the lexicographical order---especially if the column histograms have high statistical dispersion  (\S~\ref{sec:realexperiments}).

%%%%%%%%%%%%%%%%%%%%%%%%%%%%%%%%%%%%%%%%%%
\section{Related work}
%%%%%%%%%%%%%%%%%%%%%%%%%%%%%%%%%%%%%%%%%%

Many forms of compression in databases are susceptible to row reordering.
For example, to increase the compression factor, Oracle engineers~\cite{Poess:2003:DCO:1315451.1315531} recommend sorting the data
before loading it. Moreover, they recommend taking
into account the cardinality of the columns---that is, the number of distinct column values. 
Indeed, they indicate that sorting on
low-cardinality columns is more likely to increase the compression factor.
Poess and Potapov do not quantify the effect of row reordering. However, they report that compression gains on synthetic data are small (a factor of 1.4 on TPC-H) but can be much larger on real data (a factor of 3.1). The effect on performance varies from slightly longer running times to a speedup of 38\% on some queries. Loading times are doubled.

Column-oriented databases and indexes are particularly suitable
for compression. 
Column-oriented databases such as C-Store use the conventional (lexicographical)
sort to improve compression~\cite{1083658,1142548}.
Specifically, a given table is decomposed into several  overlapping projections (e.g., on columns 1,2,3 then on column 2,3,4) which are sorted and compressed. By choosing projections matching the query workload, it is possible to surpass a conventional DBMS by orders of magnitude. 
To validate their model, Stonebraker et al.\ used TPC-H with a scale factor of 10: this generated 60~million rows in the main table. They kept only attributes of type INTEGER and CHAR(1). On this data, they report a total space usage of 2\,GB compared to 2.7\,GB for an alternative column store. They have a 30\% storage advantage,  
and better performance,  
partly
because they sort their projections before compressing them.

\citeN{rlewithsorting} %Lemire and Kaser 
prove that sorting the projections on the low-cardinality column first  often maximizes compression.%~\cite{rlewithsorting}. 
They stress that picking the right column order is important as the compressibility could vary substantially (e.g., by a factor of 2 or 3). They consider various alternatives to the lexicographical order such as modular and reflected Gray-code orders or Hilbert orders, and find them ineffective.
In contrast, we propose  new heuristics  (\textsc{Vortex} and \textsc{Multiple Lists}) that can surpass the lexicographical order.
Indeed, 
when using a 
compression
technique such as Prefix coding (see \S~\ref{sec:bigsectionondict}),
Lemire and Kaser 
obtain
compression gains of more than 20\% due to sorting: using the  same compression technique, on the same data set, we report 
further gains of 21\%. \citeN{Pourabbas2012} %Pourabbas et al. 
extend the strategy by showing that columns with the same cardinality should be ordered from high skewness to low skewness.

The compression of bitmap indexes also greatly benefits from table sorting. In some experiments, the sizes of the bitmap indexes are reduced by nearly an order of magnitude~\cite{arxiv:0901.3751}. Of course, everything else being equal, smaller indexes tend to be faster.
Meanwhile, alternatives to the lexicographical order such as
\textsc{Frequent-Component}, reflected Gray-code or Hilbert orders are unhelpful on bitmap indexes~\cite{arxiv:0901.3751}. (We review an improved version of \textsc{Frequent-Component} in \S~\ref{sec:fc}.)

Still in the context of bitmap indexes,
\citeN{malik2007optimizing} %Malik and Kender 
get good compression results using a variation on the Nearest Neighbor TSP heuristic. Unfortunately, its quadratic time complexity makes the processing of large data sets difficult. 
To improve scalability, \citeN{malik2007optimizing} %Malik and Kender 
also propose a faster heuristic called  aHDO which we review in \S~\ref{sec:tsp-heuristics}.
In comparison, our novel \textsc{Multiple Lists} heuristic is also an attempt to get a more scalable  Nearest Neighbor heuristic.
Malik and Kender used
small 
data sets having between 204~rows and 34\,389~rows. All their compressed bitmap indexes were under 10\,MB\@. On their largest data set, the compression gain from aHDO was  14\% 
when compared
to the original  order. Sorting improved compression by 7\%, whereas their Nearest Neighbor TSP heuristic had the best gain at 17\%.
\citeN{pinar05} %Pinar et al.\ 
also present good compression results on bitmap indexes after reordering: on their largest data set (11\,MB), they report
using a Gray-code approach 
to get
a compression ratio of 1.64 compared to the original order. 
 Unfortunately, they do not compare with the lexicographical order.

Sometimes reordering all of the data before compression is not an option.
For example, Fusco et al.~\shortcite{netfli,springerlink:10.1007/s00778-011-0242-x} describe a system where
bitmap indexes must be compressed on-the-fly to index network traffic.
They report that their system can accommodate the insertion of more
than a million records per second.
To improve compressibility without sacrificing performance, they
cluster the rows using locality sensitive hashing~\cite{gionis1999similarity}. They report a compression factor of 2.7 due to this reordering (from 845\,MB to 314\,MB).

%%%%%%%%%%%%%%%%%%%%%%%%%%%%%%%%%%%%%%%%%%
\section{Minimizing the number of runs}
%%%%%%%%%%%%%%%%%%%%%%%%%%%%%%%%%%%%%%%%%%
\label{sec:minimizing-runs}

One of the primary benefits of column stores is the compression due to 
run-length encoding (RLE)~\cite{abadi2008column,brunoelephant,1453912}.
Moreover, the most popular bitmap-index compression techniques are variations
on RLE~\cite{DBLP:journals/tods/WuOS06}.

RLE is a compression strategy where runs of identical values are coded using
the repeated value and the length of the run. For example, the sequence
\texttt{aaaaabcc} becomes $5\times \texttt{a}, 1 \times \texttt{b}, 2 \times {c}$.
Counters may be stored
using a variable number of bits, e.g., using 
variable-byte coding~\cite{scholer2002cii,db2luw2009},
Elias delta coding~\cite{scholer2002cii}
or Golomb coding~\cite{golomb1966rle}.
Or we may store counters using a fixed number of bits for faster decoding.

RLE not only reduces the storage requirement: it also reduces the processing  time.
For example, we can compute the component-wise sum---or  indeed any $O(1)$~operation---of two RLE-compressed array in time proportional
to the total number of runs. In fact, we sometimes sacrifice compression in favor
of speed: 
\begin{itemize}
\item to help random access, we can add the row identifier to the run length and repeated value~\cite{1142548} so that $5\times \texttt{a}, 1 \times \texttt{b}, 2 \times {c}$ becomes $5\times \texttt{a} \textrm{ at 1}, 1 \times \texttt{b} \textrm{ at 6}, 2 \times {c} \textrm{ at 7}$; 
\item to simplify computations, we can forbid runs from different columns to partially overlap~\cite{brunoelephant}: unless two runs are disjoint as sets of row identifiers, then one must be a subset of the other; 
\item  to avoid the overhead of decoding too many counters, we may store single values or short runs verbatim---without
any attempt at compression~\cite{874730,DBLP:journals/tods/WuOS06}. \end{itemize}
Thus, instead of trying to model each form of RLE compression accurately, we only count the
 total number of runs (henceforth \textsc{RunCount}).

 Unfortunately, minimizing \textsc{RunCount} by row reordering is NP-hard~\cite{rlewithsorting,olke:rearranging-data-pods}.
 Therefore, we resort to heuristics. We examine many possible alternatives (see Table~\ref{tab:heur-megatable}). 
 
\begin{table}
\tbl{\label{tab:heur-megatable} Summary of heuristics considered and overall results. 
Not all methods were tested on realistic data; those not tested were 
either too inefficient for large data, or were clearly
unpromising after testing on Zipfian data. 
}{\footnotesize 
\renewcommand{\arraystretch}{1.2}
\begin{tabular}{lclp{1.23cm}p{1.15cm}} 
\multirow{2}{*}{Name}                         & \multirow{2}{*}{Reference}                     & \multirow{2}{*}{Described}                       & \multicolumn{2}{c}{Experiments} \\ %\cline{4-5}
                             &                               &                                 & Synthetic & Realistic \\ \hline 
\textsc{1-reinsertion}       &\cite{331562}                  &\S~\ref{sec:tsp-heuristics}     &  \S~\ref{sec:experiments}   &  ---\\\hdashline[1pt/1pt]
aHDO                         &\cite{malik2007optimizing}     &\S~\ref{sec:tsp-heuristics}     &    \S~\ref{sec:experiments}   & --- \\\hdashline[1pt/1pt]
\textsc{BruteForcePeephole}  & \textbf{novel}                          &\S~\ref{sec:tsp-heuristics}     &  \S~\ref{sec:experiments}   &  --- \\\hdashline[1pt/1pt]
\hspace{-1.1em}\renewcommand{\arraystretch}{0.5}
\begin{tabular}{l}\rule{0cm}{0.8em}\textsc{Farthest Insertion},\\\textsc{Nearest Insertion},\\\rule{0cm}{0.8em}\textsc{Random Insertion}\end{tabular}\renewcommand{\arraystretch}{1.0}
   &\multirow{1}{*}{\cite{rosenkrantz1977analysis}} &\multirow{1}{*}{\S~\ref{sec:tsp-heuristics}}     &   \multirow{1}{*}{\S~\ref{sec:experiments}}  &\multirow{1}{*}{---}  \\
\hdashline[1pt/1pt]
\textsc{Frequent-Component}  &\cite{arxiv:0901.3751}         &\S~\ref{sec:fc}                 &   \S~\ref{sec:experiments} & --- \\\hdashline[1pt/1pt]
\textsc{Lexicographic Sort}  &   ---                            &\S~\ref{sec:minimizing-runs}    & \S~\ref{sec:experiments}    & \S~\ref{sec:runsonrealisticdatabaset}\\\hdashline[1pt/1pt]
\textsc{Multiple Fragment}  &\cite{bentley1992fast}         &\S~\ref{sec:tsp-heuristics}     & \S~\ref{sec:experiments}   & --- \\\hdashline[1pt/1pt]
\textsc{Multiple Lists}      & \textbf{novel}                           &\S~\ref{sec:approx-nn}          &  \S~\ref{sec:experiments}   &  \S~\ref{sec:runsonrealisticdatabaset} \\\hdashline[1pt/1pt]
\textsc{Nearest Neighbor}    & \cite{bellmore1968traveling}                            &\S~\ref{sec:tsp-heuristics}     &   \S~\ref{sec:experiments} & --- \\\hdashline[1pt/1pt]
\textsc{Savings}             &\cite{clarke1964scheduling}    &\S~\ref{sec:tsp-heuristics}     &   \S~\ref{sec:experiments}   & --- \\\hdashline[1pt/1pt]
\textsc{Vortex}              & \textbf{novel}                           &\S~\ref{sec:vortex}             & \S~\ref{sec:experiments} &  \S~\ref{sec:runsonrealisticdatabaset} \\
\end{tabular}
\renewcommand{\arraystretch}{1.0}
}
\end{table}

An effective heuristic for the \textsc{RunCount} minimization problem is to sort
the rows in lexicographic order. In the lexicographic order,
the first 
component where two tuples differ ($a_j\neq b_j$ but $a_i=b_i$ for $i<j$)
determines which tuple is smaller.

There are alternatives to the lexicographical order.  A Gray code is an ordered list of tuples such that the Hamming distance between successive tuples
is one.\footnote{For a more restrictive definition, we can  replace the Hamming distance by the Lee metric~\cite{1312181}.} The Hamming distance is the number of different components between two same-length tuples, e.g., 
 \begin{eqnarray*}d(\ (a,b,y), (a,d,x)\ )=2.\end{eqnarray*}
The Hamming distance is a metric: 
 i.e., $d(x,x)=0$, $d(x,y)=d(y,x)$ and $d(x,y)+d(y,z)\geq d(x,z)$.
A Gray code over all possible tuples generates an order (henceforth a Gray-code order): $x<y$ whenever $x$ appears before $y$ in the Gray code. 
For example, we can use the mixed-radix reflected 
Gray-code order~\cite{richards1986dca,KnuthV4A} (henceforth Reflected GC). 
Consider a two-column table with column cardinalities $N_1$ and $N_2$. We label the column values from
1 to $N_1$ and 1 to $N_2$.
Starting with the tuple $(1,1)$, we generate all tuples in Reflected GC order by the following algorithm:
\begin{itemize}
\item If the first component is odd then if the second component is less than $N_2$, increment it, otherwise increment the first component.
\item If the first component is even then if the second component is greater than 1, decrement it, otherwise increment the first component.
\end{itemize}
E.g., the following list is in  Reflected GC order:
\begin{eqnarray*}(1,1), (1,2), \dots, (1,N_2), (2,N_2), (2,N_2-1), \dots, (2,1), (3,1), \dots\end{eqnarray*}
The generalization to more than two columns is straightforward. Unfortunately, the benefits of Reflected GC compared to the lexicographic order are small~\cite{malik2007optimizing,rlewithsorting,arxiv:0901.3751}.

We can bound the optimality of lexicographic orders using only the number of rows
and the cardinality of each column. Indeed, for the problem of minimizing \textsc{RunCount} by row reordering, lexicographic sorting
 is $\mu$-optimal~\cite{rlewithsorting} for a table with $n$~distinct rows and column cardinalities $N_i$ for $i=1,\ldots,c$
with 
\begin{eqnarray*}
\mu  =   \frac{\sum_{j=1}^c \min(n,\prod_{k=1}^j N_k) }{n+c-1}.
\end{eqnarray*}
 To illustrate this formula, consider
a table 
with
1~million distinct rows and four columns having cardinalities 10, 100, 1000, 10000. Then, we have
$\mu \approx 2$ which means that lexicographic sorting is 2-optimal. To apply this formula in practice,
the main difficulty might be to determine the number of distinct rows, but there are good approximation algorithms~\cite{aouiche2007cfp,1807094}. 
We can improve the bound $\mu$ slightly:

\begin{lemma}\label{lemma:omegabound}For the \textsc{RunCount} minimization problem, sorting the table lexicographically is $\omega$-optimal for
\begin{eqnarray*}
\omega & = & \frac{\sum_{i=1}^{c} n_{1,i}}{n+c-1}
\end{eqnarray*}
where $c$ is the number of columns, $n$ is the number of distinct rows, and $n_{1,j}$ is the number of distinct rows when
considering only the first $j$~columns (e.g., $n=n_{1,c}$).
\end{lemma}
\begin{proof}
Irrespective of the order of the rows, there are at least $n+c-1$~runs.
Yet, under the lexicographic order, there are no more than $n_{1,i}$~runs in the $i^{\mathrm{th}}$~column. The result follows.
\end{proof}

The bound  $\omega$ is tight.
Indeed, consider a table with $N_1, N_2, \ldots, N_c$~distinct
values in columns $1, 2, \ldots, c$ and such that it has
$n=N_1 N_2 \dots N_c$~distinct rows. The lexicographic order
will generate $ N_1 + N_1 N_2 + \dots + N_1 N_2 \dots N_c$~runs.
In the notation of Lemma~\ref{lemma:omegabound}, there are 
$\sum_{i=1}^{c} n_{1,i}$~runs. However, we can also order
the rows so that there are only $n+c-1$~runs by using the Reflected GC order. 

We have that $\omega$ is bounded by the number of columns. That is, we have that $1\leq \omega \leq  c$.
Indeed, we have that $n_{1,c}=n$ and $n_{1,i}\geq 1$ so that $\sum_{i=1}^{c} n_{1,i}\geq {n+c-1}$ and therefore  $ \omega = \frac{\sum_{i=1}^{c} n_{1,i}}{n+c-1}\geq 1$. 
We also have that $n_{1,i}\leq n$ so that  $ \sum_{i=1}^{c} n_{1,i}\leq c n \leq   c (n+c-1)$ and hence  $ \omega = \frac{\sum_{i=1}^{c} n_{1,i}}{n+c-1}\leq c$. 
 In practice, the bound $\omega$ is often  larger when $c$ is larger (see \S~\ref{sec:Datasets}).

\subsection{Run minimization and TSP}
\label{sec:runminandtsp}

There is much literature about the TSP, including approximation
algorithms and many heuristics, but our run-minimization problem is not
\emph{quite} the TSP: it more resembles 
a minimum-weight Hamiltonian path
problem because we do not complete the cycle~\cite{springerlink:10.1007/3-540-44469-6_34}.  In order to use known
TSP heuristics, we need a reduction from our problem to TSP\@.  In
particular, we
reduce the run-minimization problem to TSP over the Hamming distance $d$. 
Given the rows $r_1, r_2, \dots, r_n$, \textsc{RunCount} for $c$~columns
is given by the sum of the Hamming distance between the successive rows,
\begin{eqnarray*}
c+\sum_{i=1}^{n-1} d(r_i,r_{i+1}).
\end{eqnarray*}
Our goal is to minimize  $\sum_{i=1}^{n-1} d(r_i,r_{i+1})$.
Introduce an extra row $r_{\star}$ with the property
that $d(r_{\star},r_i) = c$ for any $i$. We can achieve
the desired result under the Hamming distance by filling in 
the row $r_{\star}$ with values that do not appear
in the other rows.
We
solve the TSP over this extended set ($r_1, \ldots, r_n,r_{\star}$) by finding a  reordering of the elements
($r'_1, \ldots, r'_n,r_{\star}$)
minimizing the sum of the Hamming distances between successive rows:
\begin{eqnarray*}
d(r'_n,r_{\star})+ d(r_{\star},r'_1) + \sum_{i=1}^{n-1} d(r'_i,r'_{i+1})  =
2c +  \sum_{i=1}^{n-1} d(r'_i,r'_{i+1}).
\end{eqnarray*} 
Any reordering minimizing $2c +  \sum_{i=1}^{n-1} d(r'_i,r'_{i+1})$
also minimizes $\sum_{i=1}^{n-1} d(r'_i,r'_{i+1})$.
Thus, we have reduced the minimization of \textsc{RunCount} by row reordering to TSP\@.
Heuristics for TSP can
now be employed for our problem---after finding a tour 
($r'_1, \ldots, r'_n,r_{\star}$), we order the table rows as
$r'_1, r'_2, \ldots , r'_n$.

Unlike the general TSP, we know of 
linear-time $c$-optimal heuristics when using the Hamming distance. 
An ordering is \emph{discriminating}~\cite{cai1995umd} if duplicates are listed consecutively. 
By constructing a hash table, we can generate a discriminating order in expected linear time. It is sufficient for $c$-optimality.

\begin{lemma}\label{lemma:discriminating}
Any discriminating row ordering is $c$-optimal for the \textsc{RunCount} minimization problem.
\end{lemma}
\begin{proof}
If $n$ is the number of distinct rows, then a discriminating row ordering
has at most $nc$~runs. Yet any ordering  generates at least $n$~runs. This proves the result. 
\end{proof}

Moreover---by the triangle inequality---there is a discriminating row order minimizing the number of runs.
In fact,  given any row ordering we can construct a
discriminating row ordering with a lesser or equal cost $\sum_{i=1}^{n-1} d(r_i,r_{i+1})$ because of the triangle inequality.
Formally, suppose that we have a non-discriminating
order $r_1, r_2, \ldots,r_n$. We can find
two identical tuples ($r_k=r_j$) separated
by at least one different tuple ($r_{k+1} \neq r_j$).
Suppose $j<n$. If we move $r_j$ between
$r_k$ and $r_{k+1}$, the cost $\sum_{i=1}^{n-1} d(r_i,r_{i+1})$
will change by 
$ d(r_{j-1},r_{j+1}) - ( d(r_{j-1},r_j)+d(r_{j},r_{j+1}))$: a  quantity 
at most
zero by the triangle inequality.
If $j=n$, the cost will change by $ - d(r_{j-1},r_j)$, another non-positive quantity.
We can repeat such moves until the new order is discriminating, which proves the result.

\subsection{TSP heuristics}
\label{sec:tsp-heuristics}

We want to solve  TSP instances with the Hamming distance.
For such metrics, one of the earliest and still unbeaten TSP heuristics is the 1.5-optimal Christofides algorithm~\cite{Christofides1976,Berman:2006:AT:1109557.1109627,Gharan:2011:RRA:2082752.2082897}.   
Unfortunately, it runs in $O(n^{2.5} (\log n)^{1.5})$ time~\cite{115366}
and even a quadratic running time would be 
prohibitive for our application.\footnote{Unless we explicitly include  the number of columns $c$ in the complexity analysis,  we consider it to be a constant.}

Thus, we consider faster alternatives~\cite{johnson2004,johnsonmcgeoch1997}. 
\begin{mylongitem}
\item Some heuristics are based on space-filling curves~\cite{platzman1989spacefilling}.
Intuitively, we want  to sort the tuples in the order in which they would appear on a curve visiting every
possible tuple. Ideally, the curve would be such that nearby points on the curve are also
nearby under the Hamming distance. In this sense, lexicographic orders---as well as the
\textsc{Vortex} order (see \S~\ref{sec:vortex})---belong to this class of heuristics even though they are not generally considered space-filling curves. Most of these heuristics run in time $O(n \log n)$. 
\item There are various  tour-construction heuristics~\cite{johnson2004}. These heuristics work by 
inserting, or appending, one tuple at a time in the solution. 
In this sense, they are greedy heuristics. They all begin with a randomly chosen starting tuple.
The simplest is \textsc{Nearest Neighbor}~\cite{bellmore1968traveling}:
  we append an available tuple, choosing one of those  
nearest to the last tuple added. 
  It runs in time $O(n^2)$
(see also Lemma~\ref{lemma:nnisnlogn}).
A variation is to   
also allow tuples to be inserted
at the beginning of the list or appended at the end~\cite{bentley1992fast}.
Another similar heuristic is \textsc{Savings}~\cite{clarke1964scheduling} which is reported to work
well with the Euclidean distance~\cite{johnson2004}.
A subclass of the tour-construction heuristics are the insertion heuristics: the selected tuple is inserted at the best possible
location in the existing tour. They differ in how they pick the tuple to be inserted:
\begin{itemize}
\item \textsc{Nearest Insertion}:  we pick a tuple nearest to a tuple in the tour. 
\item \textsc{Farthest Insertion}: we pick a tuple farthest from the tuples in the tour.
\item \textsc{Random Insertion}: we pick an available tuple at random. 
\end{itemize}
One might also pick a tuple whose cost of insertion is minimal, leading to an $O(n^2 \log n)$ heuristic. Both
this approach and \textsc{Nearest Insertion} are 2-optimal, but
the named insertion heuristics are in $O(n^2)$~\cite{rosenkrantz1977analysis}. 
There are many variations~\cite{kahng2004match}.
\item \textsc{Multiple Fragment} 
(or \textsc{Greedy}) is a bottom-up heuristic: initially, each tuple constitutes a fragment of
a tour, and fragments of tours are repeatedly merged~\cite{bentley1992fast}. 
The distance between fragments is computed by comparing the first and last tuples of both fragments. 
Under the Hamming distance, there is a  $c+1$-pass implementation strategy: first merge
fragments with Hamming distance zero, then merge fragments with Hamming distance one and so on. It runs in time $O(n^2 c^2)$. 
\item Finally, the last class of heuristics are those beginning with an existing tour. 
We continue trying to improve the tour until it is no longer possible or another stopping criteria is met.
 There are many ``tour-improvement techniques''~\cite{helsgaun2000effective,applegate2003chained}.
Several heuristics  break the tour and attempt to reconstruct a better one~\cite{croes1958method,lin1973effective,helsgaun2000effective,applegate2003chained}. 

\citeN{malik2007optimizing} %Malik and Kender 
propose the aHDO heuristic which  permutes successive tuples to improve the
solution. %Pinar et al.\ 
\citeN{pinar05} describe a similar scheme, where they consider permuting tuples that are not immediately adjacent, provided that they are not too far apart. 
\citeN{331562} %Pinar and Heath 
repeatedly remove and reinsert (henceforth
\textsc{1-Reinsertion}) a single tuple at a better location. A variation is the \textsc{BruteForcePeephole}~heuristic: divide up the table into
small non-overlapping partitions of rows, and find the optimal solution that leaves the first and last row unchanged 
(that is, we solve a Traveling Salesman Path Problem (TSPP)~\cite{springerlink:10.1007/s10107-006-0046-8}).

\end{mylongitem}

\subsection{Scaling up the heuristics}
\label{sec:horizontalpartition}

External-memory sorting is  applicable to very large tables. 
However, even one of the fastest TSP heuristics (\textsc{Nearest Neighbor}) may fail to scale. We consider several strategies to alleviate this scalability problem.

\subsubsection{Sparse graph}
\label{sec:approx-nn}

 Instead of trying to solve the problem over a dense graph, where every tuple can follow any other tuple in the tour, we may construct a sparse graph~\cite{1744275,johnson2004clb}. 
 For example, the sparse graph might be constructed by limiting each tuple to  some of its near neighbors. 
A similar approach has also been used, for example, in the design of heuristics in weighted matching~\cite{48015}
and for document identifier assignment~\cite{1772723}.
In effect, we approximate the nearest neighbors.

We consider a similar strategy. Instead of storing a sparse graph structure, we store the table in several different orders. 
We compare rows only against other rows appearing consecutively in one of the lists.
Intuitively, we consider rows appearing consecutively in sorted lists to be approximate near neighbors~\cite{276876,gionis1999similarity,258656,301325,liu2004strong,276877}.
We implemented an instance of this strategy (henceforth \textsc{Multiple Lists}).

Before we formally describe the \textsc{Multiple Lists} heuristic, consider the example
given in  Fig.~\ref{fig:multiplylinkedlistexample}. Starting 
from an initial table (Fig.~\ref{fig:multiplylinkedlistexample-a}), we sort
the table lexicographically with the first column as the primary key: this
forms a list which we represent as  solid 
edges in the graph of Fig.~\ref{fig:multiplylinkedlistexample-c}. Then, we re-sort the table, 
this time using the second column as the primary key: this forms a second list
which we represent as dotted edges in Fig.~\ref{fig:multiplylinkedlistexample-c}.
Finally, starting from one particular row (say 1,3), we can greedily pick a
nearest neighbor (say 3,3) within the newly created sparse graph. We repeat
this process iteratively (3,3 goes to 5,3 and so on) until we have the solution
given in Fig.~\ref{fig:multiplylinkedlistexample-b}.

Hence, to apply \textsc{Multiple Lists} we
pick several different ways to sort the table. For each table order, we store the result in a dynamic data structure so that
rows can be selected in  order and removed quickly. (Duplicate rows can be stored once if we keep track of their frequencies.) 
One implementation strategy uses a multiply-linked list. Let $K$ be the number of different table orders. Add to each row room for $2K$~row pointers.
First sort the row in the first order. With pointers, link the successive rows, as in a doubly-linked list---using 2~pointers per row. Resort the rows in the second
order. Link successive rows, using another 2~pointers per row. 
Continue until all $K$ orders have been processed and every row has $2K$ pointers.
Removing a row in this data structure requires the modification of up to
$4K$~pointers.

\begin{figure}
\begin{centering}
\subfloat[Initial table\label{fig:multiplylinkedlistexample-a}]{%
	\hspace*{1cm}
\begin{tabular}{cc}
1&3\\
2&1\\
2&2\\
3&3\\
4&1\\
4&2\\
5&3\\
6&1\\
6&2\\
7&4\\
8&3
\end{tabular}
\hspace*{1cm}
}
\subfloat[Possible solution\label{fig:multiplylinkedlistexample-b}]{%
\hspace*{1cm}
\begin{tabular}{cc}
1&3\\
3&3\\
5&3\\
8&3\\
7&4\\
6&2\\
6&1\\
4&1\\
4&2\\
2&2\\
2&1
\end{tabular}
\hspace*{1cm}
}
\\
\hspace{2cm}
\centering\subfloat[Multiply-linked list\label{fig:multiplylinkedlistexample-c}] {%
	\includegraphics[width=0.7\textwidth]{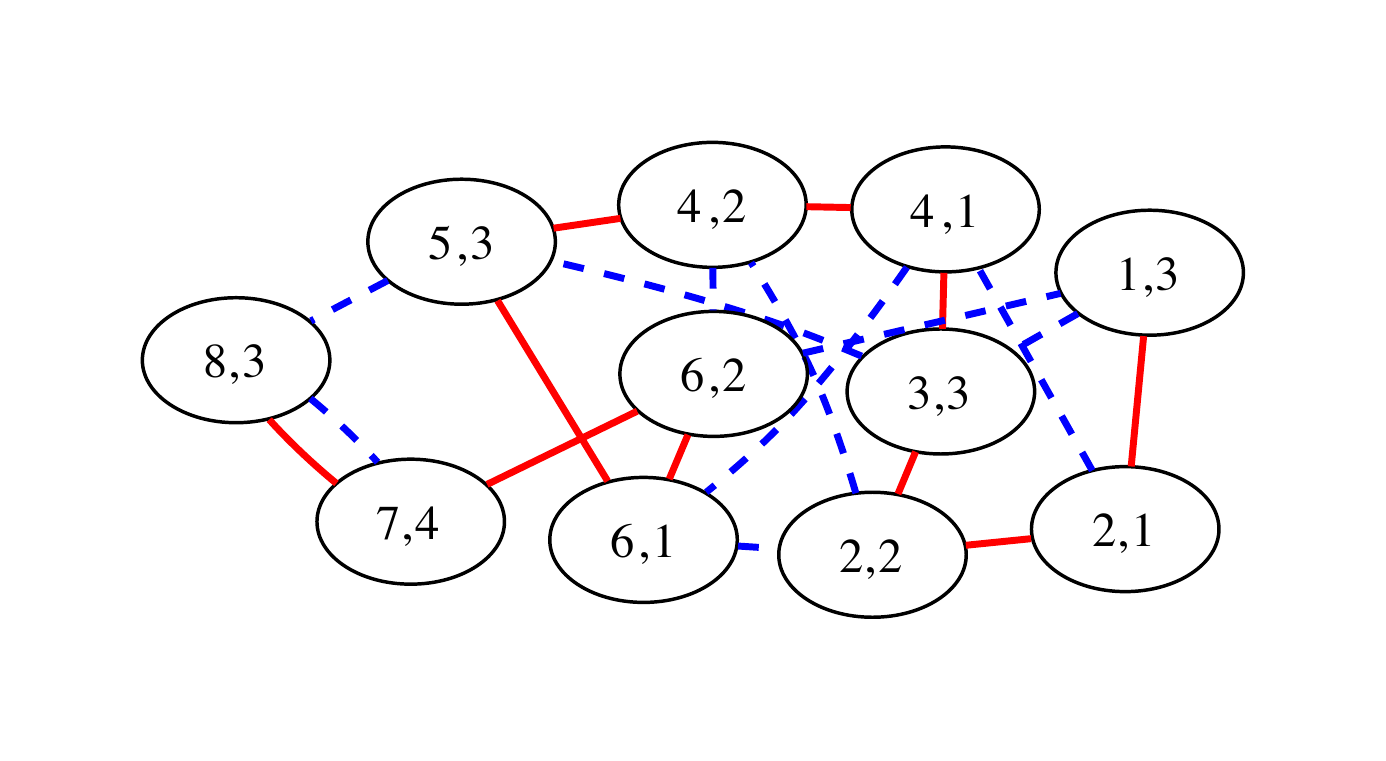}
}
\end{centering}
\caption{\label{fig:multiplylinkedlistexample} Row-reordering example
with \textsc{Multiple Lists}}
\end{figure}

For our experiments, we applied \textsc{Multiple Lists} with $K=c$ as follows. 
 First sort the table lexicographically\footnote
{Sorting with reflected Gray code yielded no appreciable improvement
on Zipfian data.
}
after ordering the columns by non-decreasing cardinalities ($N_1\leq N_2\leq \dots \leq N_c$). 
Then rotate the columns cyclically so that the first column becomes the second one, the second one becomes the third one, and the last column becomes the first: $1,2,\dots,c \rightarrow c,1,2,\dots,c-1$. Sort the table lexicographically according to this new column order. 
Repeat this process until you have $c$~different orders, each one corresponding to the same table sorted lexicographically with a different column order.

Once we have sorted the table several times,  we proceed as in \textsc{Nearest Neighbor} (see Algorithm~\ref{algo:multiplelists}): 
\begin{itemize*}
\item Start with an initially empty list $\beta$.
\item Pick a row at random, remove it from all sorted lists and add it to list $\beta$. 
\item In all sorted lists, visit all rows that are immediately adjacent (in sorted order) to the last row added. There are up to $2c$~such rows. Pick a nearest neighbor under the Hamming distance, remove it from all sorted lists and append it to     
$\beta$.
\item Repeat the last step until the sorted lists are empty. The solution is then given by list $\beta$,  which now contains all rows.
\end{itemize*}

 \textsc{Multiple Lists} runs in time $O(K c n \log n)$, or $O( c^2 n \log n)$
when we use $c$ sorted lists. 
However, we can do better in the $K=c$~case with cyclically rotated column order. First, we sort with column order $1,2,\dots, c$ in $O(c n \log n)$~time. Then, we have $N_1$~lists---one list per value in the first column---sorted in the $c,1,2,\dots,c-1$-lexicographical  order. Thus, sorting in the $c,1,2,\dots,c-1$-lexicographical order requires only $O(c n \log N_1)$~time. Similarly, sorting in the $c-1,c,1,2,\dots,c-2$-lexicographical order requires $O(c n \log N_2)$~time. And so on. Thus, the total sorting time is in 
 \begin{eqnarray*}O(c n \log n + c n \log N_1 + \dots + c n \log N_c) = O(c n \log (n N_1 N_2 \dots N_{c-1})).\end{eqnarray*} 
  We expect that $n N_1 N_2 \dots N_{c-1} \ll n ^c$ and thus \begin{eqnarray*}\log (n N_1 N_2 \dots N_{c-1})\ll c \log n\end{eqnarray*} in practice. (This approach
of reusing previous sorting orders when re-sorting is reminiscent of the Pipe~Sort algorithm~\cite{DBLP:conf/vldb/AgarwalADGNRS96}.) The overall complexity of  \textsc{Multiple Lists} is thus $O(c^2 n + c n \log (n N_1 N_2 \dots N_{c-1}))$.

\begin{algorithm}
\small  %recommended by ACM small
\begin{algorithmic}[1]
\STATE \textbf{input:}  Unsorted table $t$ with $n$ rows and $c$ columns.
\STATE \textbf{output:} a sorted table
\STATE Form
$K$~different versions of $t$, sorted differently: $t^{(1)},t^{(2)},\dots, t^{(K)}$ 

\STATE $\beta \leftarrow $ empty list
\STATE pick an element in $t^{(1)}$ randomly, add it to $\beta$ and remove it from all $t^{(i)}$'s
\WHILE{$\mathrm{size}(\beta)<n$ }
\STATE let $r$ be the latest element added to $\beta$
\STATE Given $i\in \{1,2,\dots,K\}$, there are up to two neighbors in sorted order within list $t^{(i)}$; out of up to $2K$ such neighbors, pick a nearest neighbor $r'$ to $r$ in Hamming distance.
\STATE Add $r'$ to $\beta$ and remove it from all $t^{(i)}$'s
\ENDWHILE
\STATE \textbf{return} $\beta$
\end{algorithmic}
\caption{\label{algo:multiplelists}The \textsc{Multiple Lists} heuristic}

\end{algorithm}

Although an increase
in $K$ degrades the running time, we expect each new list stored to improve the solution's quality.
In fact, the heuristic \textsc{Multiple Lists}  becomes equivalent to  \textsc{Nearest Neighbor} when we  maintain the table in all
of the $c!$~lexicographical sorting orders.
This shows the following result:

\begin{lemma}\label{lemma:nnisnlogn}The  \textsc{Nearest Neighbor} heuristic over $c$-column tables and under the Hamming distance is in $O(c\, c!\, n \log n)$.
\end{lemma}

When there are many columns ($c>4$), constructing and maintaining $c!$~lists might be prohibitive.
Through informal tests, we found that maintaining $c$~different sort orders is a good compromise.
For $c\leq 2$, \textsc{Multiple Lists} with $c$~sorted lists is already equivalent to \textsc{Nearest Neighbor}.

Unfortunately, our implementation of \textsc{Multiple Lists}
is not suited to an external-memory implementation without partitioning.

\subsubsection{Partitioning}
\label{sec:Partitioning}

Several authors
partition---or cluster---the tuples before applying a TSP
heuristic.~\cite{cesari1996divide,1744275,karp1977,johnson2004clb,10.1109/SSDM.1999.787635} Using database terminology, they partition the table
horizontally. We  explore this
approach. The content of the horizontal partitions depends on the original order of the
rows: the first rows are in the first partition and so on. 
 Hence, we are effectively considering a tour-improvement
technique: starting from an existing row ordering, we partition it
and try to reduce the number of runs within each partition.  
For example, we can partition a lexicographically sorted table
and  process each partition in main memory using expensive heuristics such as
\textsc{Nearest Neighbor} or random-access-intensive heuristics
such as \textsc{Multiple
Lists}.

We can process each partition independently: the problem is embarrassingly parallel. 
Of course, this ignores runs created at the boundary between partitions.

Sometimes, we know the final row of the previous partition. In such cases, we might choose the initial row in the current partition to have a small Hamming distance with the last
row in the previous partition. In any case, this boundary effect becomes small as the sizes of the partitions grow.

Another immediate practical benefit of the horizontal partitioning is that we have an anytime---or interruptible---algorithm~\cite{dean1988analysis}. Indeed,
we progressively improve the row ordering, but can also abort the process at any time without losing the gains achieved up to that point.

We could ensure that the number of runs is always reduced. Indeed, whenever the application of
the heuristic on the partition makes matter worse, we can  revert back to the original tuple order.
Similarly, we can try several heuristics on each partition and pick the best. And probabilistic heuristics
such as \textsc{Nearest Neighbor} can be repeated.
Moreover, we can repartition the table: e.g.,  each new partition---except the first and the last---can take half its tuples from each of the two adjacent old partitions~\cite{johnsonmcgeoch1997}.

For simplicity, 
we can use partitions having a fixed number of rows (except maybe for the last partition).
As an alternative, we could first sort the data and then create partitions based on the value of one or several columns. Thus, for example, we could ensure that all rows within a partition have the same value in one or several columns.

%%%%%%%%%%%%%%%%%%%%%%%%%%%%%%%%%%%%%%%%%%
\section{Maximizing the number of long runs}
%%%%%%%%%%%%%%%%%%%%%%%%%%%%%%%%%%%%%%%%%%
\label{sec:blockwiseruns}

It is convenient to model database compression by the number
of runs (\textsc{RunCount}). However, this model  
is clearly incomplete.
For example, there is some overhead corresponding to each run of 
values we need to code: short runs are difficult to compress.

\subsection{Heuristics for long runs}

We developed a number of row-reordering heuristics whose goal was to
produce long runs.  Two straightforward approaches did not give experimental
results that justified their costs.  

One is due to an idea of 
%Malik and Kender~
\citeN{malik2007optimizing}.
Consider the \textsc{Nearest Neighbor} TSP heuristic. Because we
use the Hamming distance, there are often several nearest neighbors for the last tuple added.
Malik and Kender proposed a modification of 
\textsc{Nearest Neighbor}  
where they determine the best
nearest neighbor based on comparisons with the previous tuples---and
not only with the latest one.
We considered many variations on this idea, and none of them
proved consistently beneficial:
e.g., when there are several nearest neighbors to the latest tuple, we tried reducing
the set to a single tuple by removing the tuples that are not also nearest to the
second last tuple, and then removing the tuples that are not also nearest to the third
last tuple, and so on.

A novel heuristic, \textsc{Iterated Matching}, was also developed
but found to be too expensive for the quality of results typically
obtained.  It is based on the observation that a weighted-matching algorithm
can form pairs of rows that have many length-two runs.  Pairs of rows can
themselves be matched into collections of  four rows with many length-four
runs, etc.    Unfortunately, known maximum-matching algorithms are expensive
and the experimental results obtained with this heuristic were not promising.
Details can be found elsewhere~\cite{LemireKaserGutarraTR12001}.

\subsection{The \textsc{Frequent-Component} order}
\label{sec:fc}

Intuitively, we would like to sort rows so that frequent values are more likely
to appear consecutively. The \textsc{Frequent-Component}~\cite{arxiv:0901.3751} order
follows this intuition.

As a preliminary step,
we compute the frequency $f(v)$ of each column value $v$ within each of the $c$ columns.
Given a tuple, we map each component to the triple (frequency, column index, column value).\footnote{% 
This differs slightly from the original presentation of the order~\cite{arxiv:0901.3751}
where the column value appeared before the column index. 
}
Thus, from the $c$~components, we derive $c$~triples (see Figs.~\ref{fig:frequentcomponentexample-a} and~\ref{fig:frequentcomponentexample-b}): 
e.g., given the tuple $(v_1, v_2, v_3)$, we get the triples $((f(v_1),1,v_1),(f(v_2),2,v_2),(f(v_3),3,v_3))$. 
We then lexicographically sort the
triples, so that the triple corresponding to a most-frequent column value appears
first---that is, we sort in reverse lexicographical order (see Fig.~\ref{fig:frequentcomponentexample-c}):
e.g.,  assuming that $f(v_3)<f(v_1)<f(v_2)$, the triples appear in sorted order as 
$((f(v_2),2,v_2),(f(v_1),1,v_1),(f(v_3),3,v_3))$.  
The new tuples are then compared against each
other lexicographically over the $3c$~values (see Fig.~\ref{fig:frequentcomponentexample-d}). When sorting, we can precompute the ordered lists of triples for
 speed. As a last step, we reconstruct the solution from the list of triples (see Fig.~\ref{fig:frequentcomponentexample-e}).

Consider a table where columns have uniform histograms: given $n$~rows and a column of
cardinality $N_i$, each value appears $n/N_i$ times. In such a case, \textsc{Frequent-Component} becomes
equivalent to the lexicographic order with the columns organized in non-decreasing cardinality.

\begin{figure}
\begin{centering}
\subfloat[Initial table\label{fig:frequentcomponentexample-a}]{%
\hspace*{0.3cm}
\begin{tabular}{cc}
1&3\\
2&1\\
2&2\\
3&3\\
4&1\\
4&2\\
5&3\\
6&1\\
6&2\\
7&4\\
8&3
\end{tabular}
\hspace*{0.3cm}
}
\subfloat[$(f(v_i),i,v_i)$\label{fig:frequentcomponentexample-b}]{%
\begin{tabular}{cc}
(1,1,1)&(4,2,3)\\
(2,1,2)&(3,2,1)\\
(2,1,2)&(3,2,2)\\
(1,1,3)&(4,2,3)\\
(2,1,4)&(3,2,1)\\
(2,1,4)&(3,2,2)\\
(1,1,5)&(4,2,3)\\
(2,1,6)&(3,2,1)\\
(2,1,6)&(3,2,2)\\
(1,1,7)&(1,2,4)\\
(1,1,8)&(4,2,3)
\end{tabular}
}
\subfloat[$(f(v_i),i,v_i)$ (sorted)\label{fig:frequentcomponentexample-c}]{%
\hspace*{0.1cm}
\begin{tabular}{cc}
(4,2,3)&(1,1,1)\\
(3,2,1)&(2,1,2)\\
(3,2,2)&(2,1,2)\\
(4,2,3)&(1,1,3)\\
(3,2,1)&(2,1,4)\\
(3,2,2)&(2,1,4)\\
(4,2,3)&(1,1,5)\\
(3,2,1)&(2,1,6)\\
(3,2,2)&(2,1,6)\\
(1,2,4)&(1,1,7)\\
(4,2,3)&(1,1,8)
\end{tabular}
\hspace*{0.1cm}
}
\subfloat[Sorted triples\label{fig:frequentcomponentexample-d}]{%
\begin{tabular}{cc}
(1,2,4)&(1,1,7)\\
(3,2,1)&(2,1,2)\\
(3,2,1)&(2,1,4)\\
(3,2,1)&(2,1,6)\\
(3,2,2)&(2,1,2)\\
(3,2,2)&(2,1,4)\\
(3,2,2)&(2,1,6)\\
(4,2,3)&(1,1,1)\\
(4,2,3)&(1,1,3)\\
(4,2,3)&(1,1,5)\\
(4,2,3)&(1,1,8)\\
\end{tabular}
}
\subfloat[Solution\label{fig:frequentcomponentexample-e}]{%
\hspace*{0.1cm}
\begin{tabular}{cc}
7&4\\
2&1\\
4&1\\
6&1\\
2&2\\
4&2\\
6&2\\
1&3\\
3&3\\
5&3\\
8&3\\
\end{tabular}
\hspace*{0.1cm}
}
\end{centering}
\caption{\label{fig:frequentcomponentexample} Row-reordering example
with \textsc{Frequent-Component}}
\end{figure}

\subsection{\textsc{Vortex}: a novel order}  \label{sec:vortex}

The \textsc{Frequent-Component} order has at least two inconveniences:
\begin{mylongitem}
\item Given $c$~columns having $N$~distinct values apiece, a  table where all
possible rows are present has
$N^c$~distinct rows. In this instance, \textsc{Frequent-Component} is equivalent to
the lexicographic order. Thus, its \textsc{RunCount} is $N^c+ N^{c-1}+\dots+1$.
Yet a better solution would be to sort the rows in a Gray-code order, generating only
$N^c + c -1$~runs. For mathematical elegance, we would rather
have Gray-code orders even though the Gray-code property may not
enhance compression in practice.
\item The  \textsc{Frequent-Component}  order requires comparisons between
the frequencies of values that are in different columns. Hence, we must
at least maintain one ordered list of values from all columns.
We would prefer a simpler alternative with less overhead.
\end{mylongitem}

Thus, to improve over \textsc{Frequent-Component}, we 
want an order that considers the frequencies of column values and yet is
a Gray-code order.
Unlike
\textsc{Frequent-Component}, we would prefer that it compare  
frequencies only between values
from the same column.

\begin{figure}\centering
\subfloat[\label{fig:lexico44}Lexicographic]{%
\hspace*{0.4cm}
\begin{tabular}{|cc|}
1 & 1\\
1 & 2\\
1 & 3\\
1 & 4\\
2 & 1\\
2 & 2\\
2 & 3\\
2 & 4\\
3 & 1\\
3 & 2\\
3 & 3\\
3 & 4\\
4 & 1\\
4 & 2\\
4 & 3\\
4 & 4\\
\end{tabular}
\hspace*{0.4cm}
}
\subfloat[\label{fig:reflected44}Reflected GC]{%
\hspace*{0.4cm}
\begin{tabular}{|cc|}
1 & 1\\
1 & 2\\
1 & 3\\
1 & 4\\
2 & 4\\
2 & 3\\
2 & 2\\
2 & 1\\
3 & 1\\
3 & 2\\
3 & 3\\
3 & 4\\
4 & 4\\
4 & 3\\
4 & 2\\
4 & 1\\
\end{tabular}
\hspace*{0.4cm}
}
\subfloat[\label{fig:vortex44}\textsc{Vortex}]{%
\hspace*{0.4cm}
\begin{tabular}{|cc|}
1 & 4\\
1 & 3\\
1 & 2\\
1 & 1\\
4 & 1\\
3 & 1\\
2 & 1\\
2 & 4\\
2 & 3\\
2 & 2\\
4 & 2\\
3 & 2\\
3 & 4\\
3 & 3\\
4 & 3\\
4 & 4\\
\end{tabular}
\hspace*{0.4cm}
}
\subfloat[\label{fig:Gray-codedcurve44}z-order with Gray codes]{%
\hspace*{0.4cm}
\begin{tabular}{|cc|}
1 & 1\\
1 & 2\\
2 & 2\\
2 & 1\\
2 & 3\\
2 & 4\\
1 & 4\\
1 & 3\\
3 & 3\\
3 & 4\\
4 & 4\\
4 & 3\\
4 & 1\\
4 & 2\\
3 & 2\\
3 & 1\\
\end{tabular}
\hspace*{0.4cm}
}
\subfloat[\label{fig:hilbert44}Hilbert]{%
\hspace*{0.4cm}
\begin{tabular}{|cc|}
1 & 1\\
2 & 1\\
2 & 2\\
1 & 2\\
1 & 3\\
1 & 4\\
2 & 4\\
2 & 3\\
3 & 3\\
3 & 4\\
4 & 4\\
4 & 3\\
4 & 2\\
3 & 2\\
3 & 1\\
4 & 1\\
\end{tabular}
\hspace*{0.4cm}
}
\caption{Two-column tables sorted in various orders }
\end{figure}

Many orders,
such  as z-orders with Gray codes~\cite{16877} 
 (see Fig.~\ref{fig:Gray-codedcurve44}) and 
Hilbert orders~\cite{hamilton2007chi,kamel1994hrt,eavis2007} (see Fig.~\ref{fig:hilbert44}), use some
form of \emph{bit interleaving}: when comparing two tuples,
we begin by comparing the most significant bits of their values before
considering the less significant bits.
Our novel \textsc{Vortex} order interleaves individual column values 
instead (see Fig.~\ref{fig:vortex44}). Informally, we describe the order
as follows: 
\begin{itemize}
\item Pick a most frequent value $x^{(1)}$ from the first column,
 select all tuples having the value $x^{(1)}$ as their first component, and put them first (see Fig.~\ref{fig:sortexample2} with $x^{(1)}=1$);
\item Consider the second column. Pick a most frequent value $y^{(1)}$. 
Among the tuples  having $x^{(1)}$ as their first 
component, select all tuples having $y^{(1)}$ as their second component and put them last.
Among the remaining tuples, select all tuples having $y^{(1)}$ as their second component, and put them first 
(see Fig.~\ref{fig:sortexample3} with $y^{(1)}=1$);
\item Repeat.
\end{itemize}
Our intuition is that, compared to bit interleaving,  this form of 
interleaving is more  likely to generate runs of identical values.
The name \textsc{Vortex} comes from the fact that initially there are long runs, followed
by shorter and shorter runs (see Figs.~\ref{fig:vortex44} and~\ref{fig:figures_VortexOrder}).

\begin{figure}
	\centering
		\includegraphics[height=3.3cm]{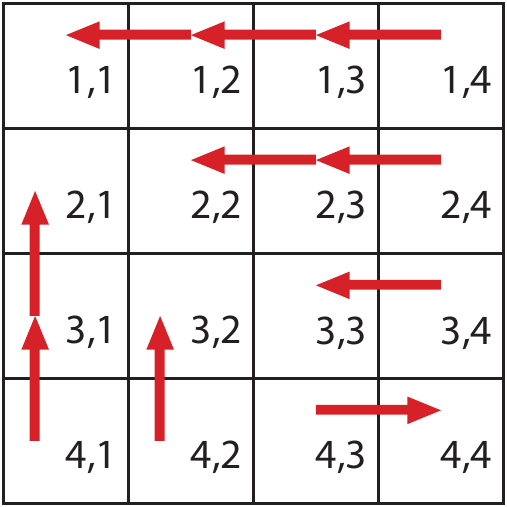}
	\caption{Graphical representation of some of the order relations between tuples in $\{1,2,3,4\}\times\{1,2,3,4\}$ under the \textsc{Vortex} order}
	\label{fig:figures_VortexOrder}
\end{figure}

\begin{figure}
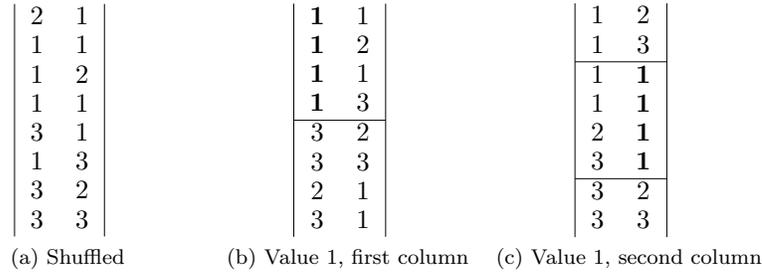
\centering
\subfloat[\label{fig:sortexample1}Shuffled]{%
\hspace*{1.2cm}\begin{tabular}{|cc|}
2 & 1\\
1 & 1\\
1 & 2\\
1 & 1 \\
3 & 1\\
1 & 3\\
3 & 2\\
3 & 3\\
\end{tabular}
\hspace*{1.2cm}}
\subfloat[\label{fig:sortexample2}Value~1, first column]{%
\hspace*{1.2cm}\begin{tabular}{|cc|}
\textbf{1} & 1 \\
\textbf{1} & 2 \\
\textbf{1} & 1 \\
\textbf{1} & 3 \\ \hline
3 & 2 \\
3 & 3 \\
2 & 1 \\
3 & 1 \\
\end{tabular}
\hspace*{1.2cm}}
\subfloat[\label{fig:sortexample3}Value~1, second column]{%
\hspace*{1.2cm}\begin{tabular}{|cc|}
1 & 2 \\
1 & 3 \\ \hline
1 & \textbf{1} \\
1 & \textbf{1} \\ 
2 & \textbf{1} \\
3 & \textbf{1} \\ \hline
3 & 2 \\
3 & 3 \\
\end{tabular}\hspace*{1.2cm}
}
\caption{First two steps in sorting with \textsc{Vortex}}
\end{figure}

\newcommand{\ltAlt}{<_{\mbox{\tiny ALT}}}
\newcommand{\ltVtx}{<_{\mbox{\tiny VORTEX}}}
\newcommand{\ltLex}{<_{\mbox{\tiny LEXICO}}}
\newcommand{\gtAlt}{>_{\mbox{\tiny ALT}}}
\newcommand{\gtVtx}{>_{\mbox{\tiny VORTEX}}}
\newcommand{\gtLex}{>_{\mbox{\tiny LEXICO}}}

To describe the \textsc{Vortex} order formally (see Algorithm~\ref{algo:vortexalgo}),
we introduce the alternating lexicographic order (henceforth \textsc{alternating}). 
Given two tuples $a$ and $b$, let $j$ be the first component 
where the two tuples differ ($a_j\neq b_j$ but $a_i=b_i$ for $i<j$),
then $a \ltAlt b$ if and only if $( a_j < b_j ) \oplus (j \textrm{ is even})$ 
where $\oplus$ is the exclusive or. (We use the convention that
components are labeled from 1 to $c$ so that the first component is
odd.)  Given a tuple $x=(x_1,x_2,\dots,x_c)$, let $T(x)$ be $(x_1,1),
(x_2,2), \ldots, (x_c,c)$ and $T'(x)$ be the list $T(x)$ sorted
lexicographically.  Then $x \ltVtx y$ if and only if $ T'(x) \ltAlt
T'(y)$.  \textsc{Vortex} generates a total order on tuples because
$T'$ is bijective and \textsc{alternating} is a total order.

We illustrate \textsc{Vortex} sorting in Fig.~\ref{fig:vortexexample}.
First, the initial table is normalized by frequency so that the
most frequent value in each column is mapped to value~1 (see Figs.~\ref{fig:vortexexample-a}, \ref{fig:vortexexample-b}, \ref{fig:vortexexample-c}).
In Fig.~\ref{fig:vortexexample-d}, we give the corresponding $T'$ values:
e.g., the row $1,4$ becomes $(1,2),(4,1)$. We then sort the $T'$ values
using the \textsc{alternating} order (see Fig.~\ref{fig:vortexexample-e}) before
finally inverting the $T'$ values to recover the rows in \textsc{Vortex} order
(see Fig.~\ref{fig:vortexexample-f}). Of course, these are not the rows of
the original table, but rather the rows of renormalized table. We could
further reverse the normalization to recover the initial table in \textsc{Vortex} order.

\begin{figure}
\begin{centering}
\subfloat[Initial table\label{fig:vortexexample-a}]{%
\hspace*{0.9cm}	
\begin{tabular}{cc}
1&3\\
2&1\\
2&2\\
3&3\\
4&1\\
4&2\\
5&3\\
6&1\\
6&2\\
7&4\\
8&3
\end{tabular}
\hspace*{0.9cm}
}
\subfloat[Renormalization\label{fig:vortexexample-b}]{%
\hspace*{0.9cm}
\begin{tabular}{c|c}
column 1& column 2 \\\hline
$2\rightarrow 1$ & $3\rightarrow 1$\\
$4\rightarrow 2$ & $1\rightarrow 2$\\
$6\rightarrow 3$ & $2\rightarrow 3$\\
$1\rightarrow 4$ & $4\rightarrow 4$\\
$3\rightarrow 5$ &\\
$5\rightarrow 6$ &\\
$7\rightarrow 7$ &\\
$8\rightarrow 8$ &\\
& \\
& \\
\end{tabular}
\hspace*{0.9cm}
}
\subfloat[Normalized table\label{fig:vortexexample-c}]{%
\hspace*{0.9cm}
\begin{tabular}{cc}
4&1\\
1&2\\
1&3\\
5&1\\
2&2\\
2&3\\
6&1\\
3&2\\
3&3\\
7&4\\
8&1
\end{tabular}
\hspace*{0.9cm}
}
\\
\subfloat[$T'$ values\label{fig:vortexexample-d}]{%
	\hspace*{1cm}
\begin{tabular}{c}
$(1,2),(4,1)$\\
$(1,1),(2,2)$\\
$(1,1),(3,2)$\\
$(1,2),(5,1)$\\
$(2,1),(2,2)$\\
$(2,1),(3,2)$\\
$(1,2),(6,1)$\\
$(2,2),(3,1)$\\
$(3,1),(3,2)$\\
$(4,2),(7,1)$\\
$(1,2),(8,1)$
\end{tabular}
\hspace*{1cm}
}
\subfloat[\textsc{alternating} order\label{fig:vortexexample-e}]{%
\hspace*{1cm}
\begin{tabular}{c}
$(1,1),(3,2)$\\
$(1,1),(2,2)$\\
$(1,2),(8,1)$\\
$(1,2),(6,1)$\\
$(1,2),(5,1)$\\
$(1,2),(4,1)$\\
$(2,1),(3,2)$\\
$(2,1),(2,2)$\\
$(2,2),(3,1)$\\
$(3,1),(3,2)$\\
$(4,2),(7,1)$\\
\end{tabular}
\hspace*{1cm}
}
\subfloat[Solution (normalized)\label{fig:vortexexample-f}]{%
\hspace*{1.8cm}
\begin{tabular}{cc}
1&3\\
1&2\\
8&1\\
6&1\\
5&1\\
4&1\\
2&3\\
2&2\\
3&2\\
3&3\\
7&4\\
\end{tabular}
\hspace*{1cm}
}
\end{centering}
\caption{\label{fig:vortexexample} Row-reordering example
with \textsc{Vortex}}
\end{figure}

Like the lexicographical order, 
the \textsc{Vortex}
order is oblivious to the column cardinalities $N_1, N_2, \ldots, N_c$:
we  only use the
content of the two tuples to determine which is smallest (see  Algorithm~\ref{algo:vortexalgo}).

%owen tried shifting things around (a small amount) so that there is not a page with just
%the algorithm and the vortex example, but it seemed there was no obvious
%solution that did not end up with a "float page".

\begin{algorithm}
\small
\begin{algorithmic}[1]
\STATE \textbf{input:} two tuples $x=(x_1,x_2,\ldots,x_c)$ and $y=(y_1,y_2,\ldots,y_c)$  
\STATE \textbf{output:} whether $x\ltVtx y$ 
\STATE $x'\leftarrow (x_1,1),(x_2,2),\ldots,(x_c,c)$ 
\STATE sort the list $x'$ lexicographically  \COMMENT{%
E.g., $(13,1), (12,2) \rightarrow (12,2), (13,1)$.}
\STATE $y'\leftarrow (y_1,1),(y_2,2),\ldots,(y_c,c)$  
\STATE sort the list $y'$ lexicographically
\FOR{$i = 1,2,\ldots,c$}
\IF{$x'_i \neq y'_i$ }
\STATE \textbf{return} $(x'_i \ltLex y'_i) \oplus (i \textrm{ is even})$
\ENDIF
\ENDFOR
\STATE \textbf{return} $\mathrm{false}$ \COMMENT{We have $x=y$.}
\end{algorithmic}
\caption{\label{algo:vortexalgo}Comparison between two tuples of integers by the \textsc{Vortex} order. We recommend organizing the columns in non-decreasing order of cardinality and the column values in non-increasing order of frequency. 
}
\end{algorithm}

Compared with \textsc{Frequent-Component},  \textsc{Vortex}
 always chooses the first (most frequent) value from
column 1, then the most frequent value from column 2, and so forth.
Indeed, we can easily show the following proposition  
using the formal definition of the \textsc{Vortex}
order:
\begin{lemma}
\label{lemma:ones-go-first}
Consider tuples with positive integer values. 
For any $1 \leq k \leq c$,
suppose $\tau$ is a tuple containing the value 1 in one of its first $k$ components,
and $\tau'$ is a tuple that does not contain the value 1 in any of its first $k$ components.
Then $\tau \ltVtx \tau '$.
\end{lemma}

\textsc{Frequent-Component} instead chooses the most frequent value overall,
then the second-most frequent value overall, and so forth, regardless of
which column contains them.  
Both \textsc{Frequent-Component} and  \textsc{Vortex} list all tuples
containing the first value consecutively. However, whereas \textsc{Vortex} also
lists the tuples containing the second value consecutively, \textsc{Frequent-Component} may fail to do so.
Thus, we expect \textsc{Vortex} to produce fewer runs among the most frequent values, compared to  \textsc{Frequent-Component}.

Whereas the Hilbert order is only a Gray code
when $N$ is a power of two, 
we want to show that
\textsc{Vortex}  is always an $N$-ary Gray code.  
That is, sorting all of
the tuples in $\{1,2,\dots,N\} \times \dots \times \{1,2,\dots,N\}= \{1,2,\dots,N\}^c$ creates
a Gray code---the Hamming distance of successive tuples is one.
We believe that of all possible orders,
we might as well pick those with the Gray-code property, everything else being equal.

Let $\mathrm{\textsc{Vortex}}({N_1,N_2,\ldots,N_c})$ be the $\prod_{i=1}^c N_i$ tuples in $\{1,2, \dots, N_1\} \times \dots \times \{1, 2, \dots, N_c\}$ sorted in the \textsc{Vortex} order.
Let
\begin{eqnarray*}\Lambda_{c,k}^{N}=\mathrm{\textsc{Vortex}}(\underbrace{N,N,\dots, N}_{c-k}, \underbrace{N-1,N-1,\ldots,N-1}_k). \end{eqnarray*}
We begin by the following technical lemmata which allow us to prove that \textsc{Vortex} is a Gray code
by induction.

\begin{lemma}\label{lemma:techlemmafromhell}
If 
$\Lambda_{c-1,k-1}^{N}$ and $\Lambda_{c,k}^{N}$ are Gray codes,
then  so is $\Lambda_{c,k-1}^{N}$ for any integers $N>1$, $c\geq 2$, $k\in\{1,\dots,c\}$.
\end{lemma}
\begin{proof}
Assume that $\Lambda_{c-1,k}^{N}$ and $\Lambda_{c,k}^{N}$ are Gray codes.
The $\Lambda_{c,k-1}^{N}$
tuples begin with the tuple 
\begin{eqnarray*}(1,\underbrace{N,\ldots,N}_{c-k},\underbrace{N-1, N-1,\ldots, N-1}_{k-1})
\end{eqnarray*} 
and they continue up to the  tuple 
\begin{eqnarray*}(1,1,\underbrace{N,\ldots,N}_{c-k-1},\underbrace{N-1,N-1,\ldots,N-1}_{k-1}).
\end{eqnarray*} 
Except for the first column which has a fixed value (one), these tuples are in reverse $\Lambda_{c-1,k-1}^{N}$ order, so they form a Gray code.
The next tuple in $\Lambda_{c,k-1}^{N}$ is 
\begin{eqnarray*}(N,1,\underbrace{N,\ldots,N}_{c-k-1},\underbrace{N-1,N-1,\ldots,N-1}_{k-1})
\end{eqnarray*} and the following tuples are in an order
equivalent to $\Lambda_{c,k}^N$, except that we must consider the first column as the last and decrement its values by one, while the second column is considered the first, the third column the second, and so on. 
The proof concludes.
\end{proof}

\begin{lemma}\label{lemma:nequal2isgraycode}
For all $c\geq 1$ and all $k\in \{0,1,\ldots, c\}$,  $\Lambda_{c,k}^{2}$ is a Gray code. 
\end{lemma}
\begin{proof}
We have that $\Lambda_{c,0}^{2}$ is a reflected Gray code with column values 1 and 2: e.g., for $c=3$, we have the order
\begin{eqnarray*}\Lambda_{3,0}^{2} & = &
\begin{matrix}
1, 2, 2\\
1, 2, 1\\
1, 1, 1\\
1, 1, 2\\
2, 1, 2\\
2, 1, 1 \\
2, 2, 1\\
2, 2, 2\\
\end{matrix}.
\end{eqnarray*}
Thus, we have that  $\Lambda_{c,0}^{2}$ is always a Gray code. Any column with cardinality one can be trivially
discarded.
Hence, we have that $\Lambda_{c,k}^{2}$ is always a Gray code for all $k\in \{0,1,\ldots, c\}$, proving the lemma.
\end{proof}

\begin{lemma}\label{lemma:onedisgraycode}
For any value of $N$ and for $k=0$ or $k=1$,  $\Lambda_{1,k}^{N}$ is a Gray code.

\end{lemma}
\begin{proof}
We have that  $\Lambda_{1,0}^{N}$ is a Gray code for any value of $N$ because one-component tuples are always
Gray codes. It immediately follows that  $\Lambda_{1,1}^{N}$ is also a Gray code for any $N$ since, by 
definition, $\Lambda_{1,1}^{N}=\Lambda_{1,0}^{N-1}$.
\end{proof}

\begin{proposition}\label{prop:vortexisagraycode}The \textsc{Vortex} order is an $N$-ary Gray code. 
\end{proposition}
\begin{proof}
The proof uses a multiple induction argument, which we express  as 
pseudocode. At the end of each pass through the loop on $N$, we have that  $\Lambda_{c,k}^{N'}$ is a Gray code
for all $N'\leq N$, all $c\in \{ 2,\dots,\mathcal{C}\}$ and all $k\in \{ 0, 1,2,\ldots, c\}$.  
The induction begins from the cases  where there are only two values per column ($N=2$, see Lemma~\ref{lemma:nequal2isgraycode}) or only one column ($c=1$, Lemma~\ref{lemma:onedisgraycode}). See Fig.~\ref{fig:illustrationofthebignastyproof}
for an illustration.

\begin{quote} %force indenting as desired for inline display of "programs", rewrite rules etc.  Rules for algorithms don't seem to discuss inline algorithms specifically.
\small %per acmsmall recommendation for algorithms.
\begin{algorithmic}[1]
\FOR{$N=3,4,\ldots,\mathcal{N}$}
\STATE We have that $\Lambda_{c,k}^{N-1}$ is a Gray code for all  $c\in \{ 2,3, \dots,\mathcal{C}\}$ and 
also for %verbose to make a line prettier
all $k\in \{ 0, 1,2,\ldots, c\}$. \COMMENT{For $N>3$, this is true from the previous pass in
the loop on $N$. For $N=3$, it follows from Lemma~\ref{lemma:nequal2isgraycode}.}\label{line:longcomment}
\FOR{$c=2,\dots,\mathcal{C}$\label{line:looponc}}
\STATE By line~\ref{line:longcomment},  $\Lambda_{c,0}^{N-1}$ is a Gray code, and by definition $\Lambda_{c,c}^{N}=\Lambda_{c,0}^{N-1}$.\label{line:linefromhell}
\FOR{$k=c, c-1, \dots, 1$}
\STATE (1) $\Lambda_{c,k}^{N}$ is a Gray code \COMMENT{when $k=c$ it follows by line~\ref{line:linefromhell} and otherwise by line~\ref{line:bigconclusion} from the previous pass of %was in, but owen noticed "of" used below; hope semantics does not change
 the loop on $k$}
\label{line:needslemmaOne}
\STATE (2) $\Lambda_{c-1,k-1}^{N}$ is a Gray code \COMMENT{when $c=2$, it follows by Lemma~\ref{lemma:onedisgraycode}, and for $c>2$ it follows from the previous pass of the loop on $c$}
\label{line:needslemmaTwo}
\STATE (1) + (2) $\Rightarrow$  $\Lambda_{c,k-1}^{N}$ is a Gray code by Lemma~\ref{lemma:techlemmafromhell}.\label{line:bigconclusion}
\ENDFOR
\STATE Hence, %we have --playing with lines
 $\Lambda_{c,0}^{N}$ is a Gray code.  (%As well as 
And also %playing with lines
$\Lambda_{c,k}^{N}$  for  all $k\in \{ 0, 1, 2, \ldots, c-1\}$.)
\ENDFOR
\ENDFOR
\end{algorithmic}
\end{quote}

\begin{figure}
\centering
\begin{tabular}{c|cccc}
          &   $c=1$  & $c=2$ & $\cdots$ & $\mathcal{C}$ \\ \hline 
$N=2$       &       \multicolumn{4}{c}{\small  Lemma~\ref{lemma:nequal2isgraycode}}           \\   
$N=3$      &   \multirow{2}{*}{\begin{turn}{90}\small Lemma~\ref{lemma:onedisgraycode}\end{turn}}     &    $\rightarrow$    & $\rightarrow$   &    \\
\multirow{2}{*}{$\vdots$} &        & \multirow{2}{*}{$\vdots$}    &  \multirow{2}{*}{$\vdots$} & \\ 
& & & \\
& & & \\
$\mathcal{N}$ &        &    $\rightarrow$   &  $\rightarrow$    & \\ 
\end{tabular}
\caption{\label{fig:illustrationofthebignastyproof}Illustration of the basis for the induction argument in the proof of Proposition~\ref{prop:vortexisagraycode}}
\end{figure}

The integers $\mathcal{N}$ and $\mathcal{C}$ can be arbitrarily large.
Thus, the pseudocode shows that $\Lambda_{c,k}^{N}$ is always a Gray code which proves
that \textsc{Vortex} is an $N$-ary Gray code for any number of columns and any value of $N$.
\end{proof}

\section{Experiments: TSP and synthetic data}
\label{sec:experiments}

We present two groups of experiments.  In this section, we use synthetic data to compare
many heuristics for minimizing \textsc{RunCount}.
Because the minimization of \textsc{RunCount}
reduces to the TSP over the Hamming distance, we effectively 
assess TSP heuristics. 
Two heuristics stand out, and then in \S~\ref{sec:realexperiments} these two heuristics are assessed
in more realistic settings  with actual database compression techniques and 
large tables. The Java source code necessary to reproduce our experiments is at~\url{http://code.google.com/p/rowreorderingjavalibrary/}.

The aHDO scheme 
runs a complete pass through the table, trying to permute successive
rows. If a pair of rows is permuted, we run through another pass over
the entire table. We repeat until no improvement is possible.
In contrast, the more expensive tour-improvement heuristics (\textsc{BruteForcePeephole} 
 and \textsc{1-reinsertion})
do a single pass through the table. The \textsc{BruteForcePeephole} is applied on successive
blocks of 8~rows.

\subsection{Reducing the number of runs on Zipfian tables}
\label{sec:results-zipf}

Zipfian distributions are commonly used~\cite{eavis2007,houkjoer2006simple,Gray:1994:QGB:191843.191886} to model value distributions
in databases:
within a column, the frequency of the $i^{\mathrm{th}}$~value is 
proportional to 
$1/i$. If a table has $n$~rows, we allow each column to have
$n$~possible distinct values, not all of which will usually appear.
We generated tables with 8\,192--1\,048\,576~rows, using four Zipfian-distributed columns that were generated
independently. Applying Lemma~\ref{lemma:omegabound} to these tables, 
the lexicographical order is $\omega$-optimal for
$\omega \approx 3$.

\begin{table}\tbl{\label{table:owenstable}Relative reduction in \textsc{RunCount} compared to lexicographical sort (set at 1.0), for Zipfian tables ($c=4$) }{%
\begin{tabular}{l|ccc}
 & \multicolumn{3}{c}{relative \textsc{RunCount} reduction} \\
 & $n=$8\,192 &	$n=$131\,072 &	$n=1\,048\,576$\\\hline
\textsc{Lexicographic Sort} & 1.000 &	1.000	& 1.00	\\
\textsc{Multiple Lists} & 1.167	& 1.188	& \textbf{1.204}	\\
\textsc{Vortex}	& 1.154	& 1.186& 	1.203\\
\textsc{Frequent-Component}	& 1.151	& 1.186	& 1.203	\\
\textsc{Nearest Neighbor}&  	1.223&	1.242 &	\\	
\textsc{Savings}&  	\textbf{1.225}	& \textbf{1.243} &	\\	
\textsc{Multiple Fragment}	&  1.219&	1.232 &	\\	
\textsc{Farthest Insertion}		&1.187 & &	\\		
\textsc{Nearest Insertion}	& 1.214	 & & \\		
 \textsc{Random Insertion}& 1.201		& &	\\		
\textsc{Lexicographic Sort}+\textsc{1-reinsertion}&1.171		& &		\\	
\textsc{Vortex}+\textsc{1-reinsertion}&1.193		& &	\\	
\textsc{Frequent-Component}+\textsc{1-reinsertion}& 1.191& &	\\	
\end{tabular}
}
\end{table}

The \textsc{RunCount} results are presented in Table~\ref{table:owenstable}.
We present relative \textsc{RunCount} reduction values: a value of 1.2 means that
the heuristic is 20\% better than lexicographic sort.
For some less scalable heuristics, we only give results for small tables. Moreover, we omit results for
aHDO~\cite{malik2007optimizing} and \textsc{BruteForcePeephole} (with partitions of eight rows) because
these tour-improvement heuristics failed to improve any tour by more than 1\%. Because 
\textsc{BruteForcePeephole}  fails, we conclude that
all heuristics we consider are ``locally almost optimal'' because you cannot improve them appreciably by 
optimizing subsets of 8~consecutive rows.

Both \textsc{Frequent-Component} and \textsc{Vortex} are better than the lexicographic order. 
The \textsc{Multiple Lists} heuristic is even better but slower. Other heuristics such as \textsc{Nearest Neighbor}, 
\textsc{Savings},  \textsc{Multiple Fragment} and the insertion heuristics  are even better, but with worse running time complexity. 
For all heuristics, the run-reduction efficiency grows
with the number of rows. 

\subsection{Reducing the number of runs on uniformly distributed tables}
\label{sec:results-uniform}

We also ran experiments using uniformly distributed tables, generating
$n$-row tables where each column has $n$~possible values.
Any value within the table 
can take one of $n$~distinct values with probability $1/n$.\footnote{We abuse the terminology slightly by referring to these tables as uniformly distributed:
the models used to generate the tables are uniformly distributed, but the data of the generated
tables may not have uniform histograms. }
For these tables, 
the lexicographical order is 3.6-optimal according to Lemma~\ref{lemma:omegabound}.

Table~\ref{table:owenstableuni} summarizes the results.
The most striking difference with Zipfian tables is that the efficiency
of most heuristics drops drastically. Whereas \textsc{Vortex}
and \textsc{Frequent-Component}  are 20\% superior to
the lexicographical order on Zipfian data, they are barely better (2\%)
on uniformly distributed data. Moreover, we fail to see improved gains
as the tables grow larger---unlike the Zipfian case.
Intuitively, a uniform model implies
that there are fewer opportunities to create long runs of identical
values in several columns, when compared to a Zipfian model. This probably explains the poor
performance of  \textsc{Vortex}
and \textsc{Frequent-Component}.

We were surprised by how well \textsc{Multiple Lists} performed on uniformly
distributed data, even for small $n$. It
fared better than most heuristics, including \textsc{Nearest Neighbor}, by
beating lexicographic sort by 13\%.   Since \textsc{Multiple Lists} is
a variant of the greedy \textsc{Nearest Neighbor} that considers a subset
of the possible neighbors, we see that more choice does not necessarily
give better results. 
\textsc{Multiple Lists}  is a good choice for this problem.

\begin{table}
\tbl{\label{table:owenstableuni}Relative reduction in \textsc{RunCount} compared to lexicographical sort (set at 1.0), for uniformly distributed tables ($c=4$) }{%
\begin{tabular}{l|ccc}%
 & \multicolumn{3}{c}{relative \textsc{RunCount} reduction} \\
 & $n=$8\,192 &	$n=$131\,072 &	$n=1\,048\,576$\\\hline
\textsc{Lexicographic Sort} & 1.000 &	1.000	& 1.00	\\
\textsc{Multiple Lists} & 1.127	& \textrm{1.128}	& \textbf{1.128}	\\
\textsc{Vortex}	& 1.020	& 1.020& 	1.021\\
\textsc{Frequent-Component}	& 1.022	& 1.023	& 1.023	\\
\textsc{Nearest Neighbor}&  	1.122&	1.123 &	\\	
\textsc{Savings}&  	1.122	& 1.123 &	\\
\textsc{Multiple Fragment}	&  \textbf{1.133}&	\textbf{1.133} &	\\	
\textsc{Farthest Insertion}		&1.075 & &	\\		
\textsc{Nearest Insertion}	& \textrm{1.129}	 & & \\		
 \textsc{Random Insertion}& 1.103		& &	\\		
\textsc{Lexicographic Sort}+\textsc{1-reinsertion}&1.092		& &		\\	
\textsc{Vortex}+\textsc{1-reinsertion}&1.094		& &	\\	
\textsc{Frequent-Component}+\textsc{1-reinsertion}& 1.080& &	\\	
\end{tabular}
}
\end{table}

\subsection{Discussion}

For these in-memory synthetic data sets, we find that the good
heuristics to minimize \textsc{RunCount}---that is, to solve the TSP
under the Hamming distance---are \textsc{Vortex} and \textsc{Multiple
  Lists}. They are both reasonably scalable in the number of rows ($O(n \log n)$) and they perform
well as TSP heuristics. 

\textsc{Frequent-Component} would be another worthy alternative, but it is harder to implement as efficiently as  \textsc{Vortex}.
Similarly, we found that \textsc{Savings} and \textsc{Multiple Fragment} could be superior TSP heuristics for Zipfian and uniformly distributed tables, but they scale poorly with respect to the number of rows: they have a quadratic  running time ($O(n^2)$). For small tables ($n=131\,072$), they were three and four orders of magnitude slower than \textsc{Multiple Lists} in our tests.

\section{Experiments with realistic tables}
\label{sec:realexperiments}
 
 Minimizing \textsc{RunCount}
 on synthetic data might have theoretical appeal 
and be somewhat applicable to many applications.  However,
 we also want to determine whether row-reordering heuristics 
can 
specifically
improve table compression, and therefore database performance,
 on realistic examples.  
This requires that we move from general models such as 
\textsc{RunCount}  to size measurements
based on specific database-compression techniques.  It also requires  large real data sets.

\subsection{Realistic column storage}
\label{sec:realistic-column-storage}
We use conventional dictionary coding~\cite{springerlink:10.1007/978-3-642-15105-7_10,1559877,Poess:2003:DCO:1315451.1315531}
prior to row-reordering and compression.
That is,  we map column values bijectively
to 32-bit integers in $[0,N)$ where $N$ is the number of distinct column
 values.\footnote{%
We do not store actual column values (such as strings): their
compression is outside our scope.}
We map the most frequent values to the smallest  integers~\cite{rlewithsorting}.

We compress tables using five database compression schemes:
Sparse, Indirect and Prefix coding as well as a fast
variation on Lempel-Ziv
and RLE\@.
We compare the compression ratio of each technique under
row reordering. For simplicity, we select only one compression
scheme at a time: all columns are compressed using the
same technique.

%%%%%%%%%%%%%%%%%%%%%%%%%%%%%%%%%%%%%%
\subsubsection{Block-wise compression}
\label{sec:bigsectionondict}

The  SAP~NetWeaver 
platform~\cite{springerlink:10.1007/978-3-642-15105-7_10} uses 
three compression techniques: Indirect, Sparse and Prefix coding. We implemented them and set the block size to $p=128$~values.

%could be recrafted to be acm \begin{description}
\begin{mylongitem}
\item 
\emph{Indirect coding} is a block-aware dictionary technique. For each column and each block,
we build a list of the $N'$~values encountered. Typically, $N' \ll N$, where $N$ is the total number
of distinct values in the column.  Column values are then mapped
to integers in $[0,N')$ and packed using $\lceil \log N \rceil$~bits each~\cite{ng1997block,1559877}.
Of course, we must also store the actual values of the $N'$~codes used
for each block and column. Thus, whereas dictionary coding requires $p \lceil \log N \rceil$~bits to store a
column, Indirect coding requires $N'\lceil \log N \rceil+ p \lceil \log N' \rceil$~bits---plus the
small 
overhead of storing the value of $N'$. In
the worst case, the storage cost of indirect coding is twice as large as the storage cost
of conventional dictionary coding. 
However, whenever 
$N'$ is small enough,
indirect coding is preferable. 
Indirect coding is related to the block-wise value coding used by Oracle~\cite{Poess:2003:DCO:1315451.1315531}.

\item
\emph{Sparse coding} stores the most frequent value using  an $p$-bit bitmap to indicate 
where this most frequent value appears. Other values are stored 
using $\lceil \log N \rceil$~bits.
If the most frequent value appears $\zeta$~times,
then the total storage cost is $(p-\zeta+1) \lceil \log N \rceil+p$~bits.

\item

\emph{Prefix coding} begins by counting how many times
the first value repeats at the beginning of the block;  this counter is stored first
along with the value being repeated.  
Then all other values are packed. In the worst case, the first value is not repeated, and
Prefix coding wastes $\lceil \log p \rceil$~bits compared to conventional dictionary coding.
Because
it counts the length of a run, 
Prefix coding can be considered a form of  RLE (\S~\ref{sec:bigsectrle}).
\end{mylongitem}

%%%%%%%%%%%%%%%%%%%%%%%%%%%%%%%%%%%%%%
\subsubsection{Lempel-Ziv-Oberhumer}
\label{sec:lzo}
Lempel-Ziv compression~\cite{ziv1978compression} compresses data by replacing sequences of characters 
by references
to 
previously encountered sequences. 
The Lempel-Ziv-Oberhumer (LZO)  library~\cite{lzolib} implements fast versions of the conventional Lempel-Ziv compression 
technique. 
%Adabi et al.~
\citeN{1142548} evaluated several alternatives including  Huffman and Arithmetic encoding,
but found that compared with LZO, they were all
too expensive for databases.

If a long data array is made of repeated characters---e.g., \texttt{aaaaa}---or repeated
short sequences---e.g., \texttt{ababab}---we expect most Lempel-Ziv compression techniques
to generate a compressed output that grows logarithmically with the size of the input.
It appears to be the case with the LZO library: its LZO1X CODEC uses 16~bytes to code 64~identical 32-bit integers,
17~bytes to code 128~identical integers, and 19~bytes for 256~integers. For our tests, we used LZO~version~2.05.

%%%%%%%%%%%%%%%%%%%%%%%%%%%%%%%%%%%%%%
\subsubsection{Run-Length Encoding}
\label{sec:bigsectrle}

We implemented RLE by storing each run as a triple~\cite{1083658}:
value, starting point and length of the run. Values are packed using 
$\lceil \log N \rceil$~bits whereas both the starting point and the length
of the run are stored using $\lceil \log n \rceil$~bits, where $n$ is the
number of rows.

\subsection{Data sets}
\label{sec:Datasets}
\begin{table}\tbl{\label{tab:caractDataSet}Characteristics of data sets used with bounds on the optimality of lexicographical order ($\omega$), and a measure of statistical deviation ($p_0$).}{%
    \begin{small}
    \begin{tabular}{lrrrrrr} %
     &  rows & distinct rows & cols & \rule{0mm}{1.1em} $\sum_i N_i$ & $\omega$ & $p_0$\\ \hline
    \multicolumn{1}{l}{\textbf{Census1881} }        & 4\,277\,807  & 4\,262\,238  & 7 & 343\,422   & 2.9 & 0.17\\
    \multicolumn{1}{l}{\textbf{Census-Income} }         & 199\,523 & 196\,294   & 42 & 103\,419  & 12 & 0.65\\
    \multicolumn{1}{l}{\textbf{Wikileaks} }        & 1\,178\,559  & 1\,147\,059  & 4 & 10\,513   & 1.3 & 0.04\\
    \multicolumn{1}{l}{\textbf{SSB} (DBGEN) }        & 240\,012\,290  & 240\,012\,290   & 15 & 8\,874\,195   & 8.0 & 0.10\\
    \multicolumn{1}{l}{\textbf{Weather} }        & 124\,164\,607  & 124\,164\,371  & 19 & 52\,369   & 4.5 &  0.36\\
    \multicolumn{1}{l}{\textbf{USCensus2000} }        & 37\,019\,068  & 22\,493\,432  & 10 & 2\,774\,239   & 3.4 & 0.54\\
    \end{tabular}
     \end{small}
     }
\end{table}

We selected realistic data sets (see Table~\ref{tab:caractDataSet}): Census1881~\cite{rlewithsorting,prdh:website}, 
Census-Income~\cite{MLRepository}, 
a Wikileaks-related data set,
the fact table from the Star Schema Benchmark (SSB)~\cite{1692916}
 and Weather~\cite{hahn:weatherbench}.  Census1881 comes from the Canadian census of 1881: it is over 305\,MB and over 4~million records.
Census-Income is the smallest data set with 100\,MB and 199\,523~records. However, it has 42~columns 
and one column has a very high relative cardinality (99\,800 distinct values).
The Wikileaks table was created from a public repository published by Google\footnote{\url{http://www.google.com/fusiontables/DataSource?dsrcid=224453}} and it contains the non-classified metadata related to 
leaked diplomatic cables. We extracted 4~columns: year, time, place and descriptive code.
It has 1\,178\,559~records.
We generated the SSB fact table using a version of the DBGEN software
modified by O'Neil.\footnote{\url{http://www.cs.umb.edu/~poneil/publist.html}}  We used a 
scale factor of 40 to generate it: that is, we used  command  
\texttt{dbgen -s 40 -T l}. 
It is 20\,GB and  
includes 240~million~rows. 
 The largest non-synthetic data set (Weather)
is 9\,GB\@.
It consists of 124 million surface synoptic weather reports from land stations for the 10-year period from December 1981 through November 1991. 
We also extracted a table from the US Census of 2000~\cite{uscensus2000file3} (henceforth USCensus2000). We used attributes 5 to 15  
from summary file~3. The resulting table has 37~million rows. 

The column cardinalities for Census1881 range from 138 to 152\,882 and from 2~to 99\,800~for Census-Income. Our Wikileaks table has column cardinalities 273, 1\,440, 3\,935 and 4\,865.
The SSB table has column cardinalities ranging from 1~to 
6\,084\,386. (The fact table has a column with a single value in
it: zero.)
For Weather, the column cardinalities range from 2 to 28\,767. 
Attribute cardinalities for USCensus2000 vary between 130\,001 to 534\,896.

For each data set, we give the suboptimality factor $\omega$ from Lemma~\ref{lemma:omegabound} in Table~\ref{tab:caractDataSet}. Wikileaks has
the lowest factor ($\omega=1.3$) followed by Census1881 ($\omega=2.9$)
whereas Census-Income has the largest one ($\omega=12.4$). Correspondingly, Wikileaks and Census1881 have the fewest columns (4 and 7), and Census-Income has
the most~(42).

We also provide a simple measure of the statistical dispersion of the frequency of the values.
 For column $i$, we find  
a 
most frequent value $v_i$, 
and we determine what fraction of this column's values are $v_i$.
Averaging the fractions, we have our measure $p_0 = \sum_{i=1}^c f(v_i)/nc$, where our table has $n$ rows and $c$ columns and
$f(v_i)$ is the number of times $v_i$ occurs in its column. 
As an example, consider the table in Fig.~\ref{fig:multiplylinkedlistexample-a}. In the first column, the value `6' is a most frequent value and it appears twice in 11~tuples. In the second column, the value `3' is most frequent, and it appears 4~times. In this example, we have $p_0 = \frac{2+4}{2\times 11}\approx 0.27 $.
In general, we have that $p_0$ ranges between 0 and 1. For uniformly distributed tables having high cardinality columns, the fraction $p_0$ is near zero. When there is high statistical dispersion of the frequencies, we expect that the most frequent values have high frequencies and so $p_0$ could be close to 1. One  
benefit of this measure is that it can be computed efficiently. By this measure, we have the highest statistical dispersion in Census-Income, Weather and USCensus2000.

\subsection{Implementation}
\label{sec:cpp-heuristic-details}
 
 Because we must be able to process
 large files in a reasonable time, we selected \textsc{Vortex}
 as one promising row-reordering heuristic.
We implemented \textsc{Vortex} in  memory-conscious manner: 
intermediate
tuples required for a comparison between rows are repeatedly
built on-the-fly.

\begin{figure}
\subfloat[Running time]{%
\includegraphics[width=0.47\textwidth]{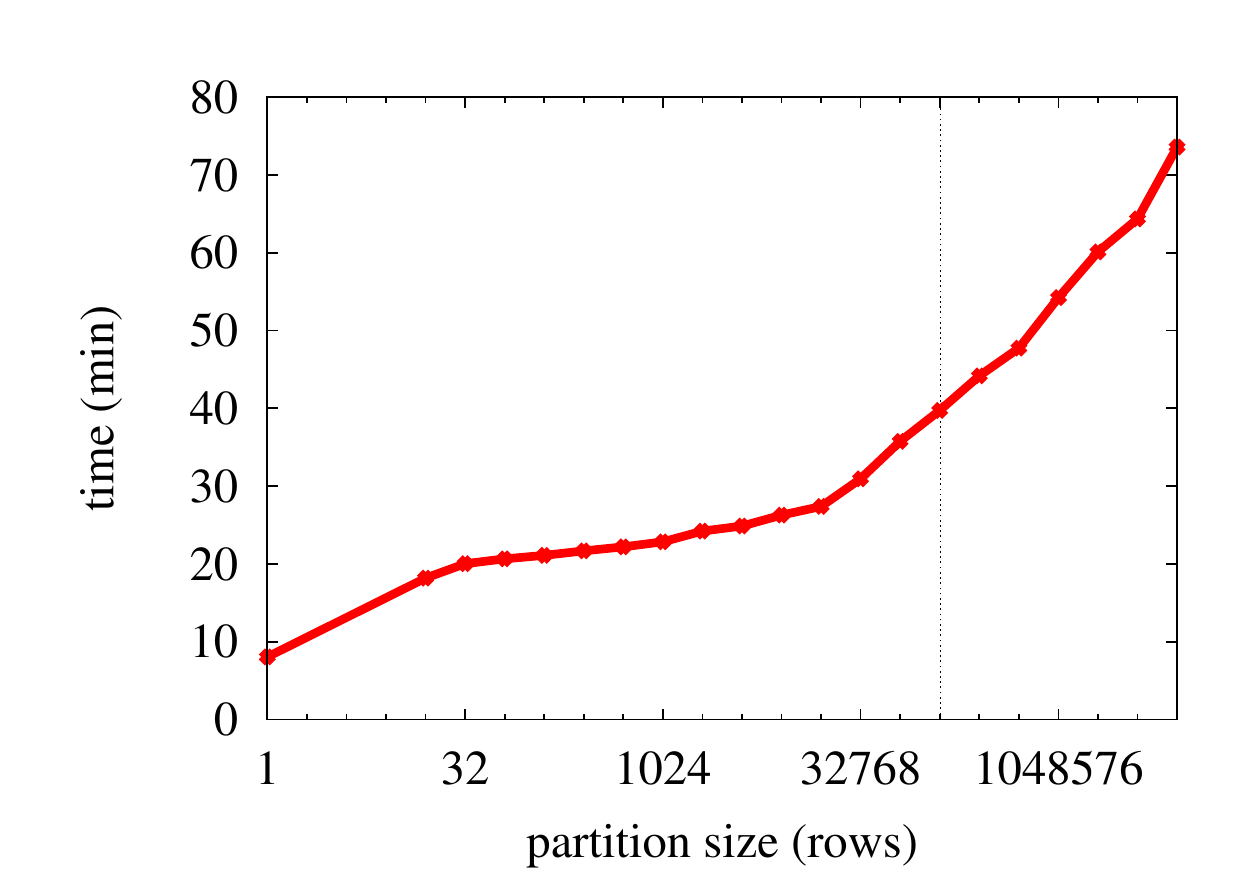}
}
\subfloat[Compressed data size]{%
\includegraphics[width=0.47\textwidth]{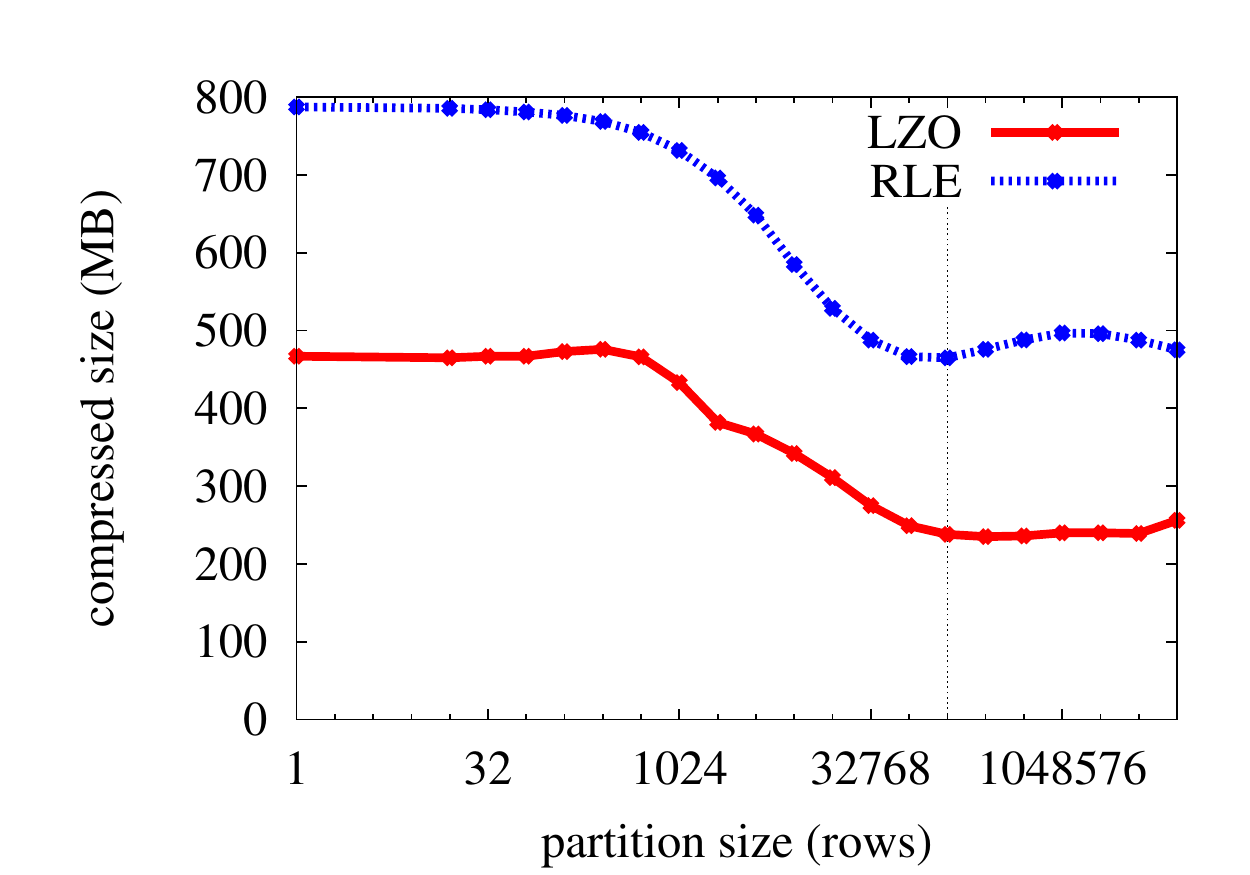}
}
\caption{\label{fig:scalingmultiplelist}The \textsc{Multiple Lists} heuristic applied on each partition of a lexicographically sorted  Weather table for various partition sizes. When partitions have size 1, 
we recover the lexicographical order. We refer to the case where the partition size is set to 131\,072~rows  as 
$\textsc{Multiple Lists}^{\star}$ (indicated by a dashed vertical line in the plots).}
\end{figure}
 
Prior to sorting lexicographically, we reorder columns in order of
non-decreasing cardinality~\cite{rlewithsorting}: in all cases this was
preferable to ordering the columns in decreasing cardinality. For \textsc{Vortex},
in all cases, there was no significant difference ($<1\%$) between ordering the columns in
increasing or decreasing order of cardinality. Effectively, \textsc{Vortex} does not
favor a particular column order.
% new last minute sentence for the TR
 (A related property can be 
formally proved~\cite{LemireKaserGutarraTR12001}.)

  We used 
 \textsc{Multiple Lists} on partitions of
 a lexicographically sorted table (see \S~\ref{sec:horizontalpartition}).\footnote{When reporting the running time, we include the time required to sort the table lexicographically.} 
Fig.~\ref{fig:scalingmultiplelist} shows the effect of the
partition size on both the running time and the data compression.
Though larger partitions may improve compression, they also require more time. As a  default, we  chose  partitions of 131\,072~rows.
 Henceforth, we refer to this heuristic as 
$\textsc{Multiple Lists}^{\star}$.

We compiled our C++
software under GNU~GCC~4.5.2 using the -O3 flag.
The C++ source code is at \url{http://code.google.com/p/rowreorderingcpplibrary/}.
We ran our experiments on an Intel Core~i7~2600 processor with 16\,GB of RAM using Linux.
All data was stored on disk, before and after the compression.
 We used a 1\,TB SATA hard drive with an estimated sustained reading speed of 115\,MB/s. 
We used an external-memory
sort that we implemented in the conventional manner~\cite{KnuthV3E3}:
we load partitions of rows, which we sort in RAM and write back in the file;
we then merge these sorted partitions using a priority queue. 
Our code is sequential.

Our algorithms are scalable. Indeed, both \textsc{Vortex} and lexicographical sorting rely on the standard external-memory sorting algorithm. The only difference between \textsc{Vortex} and lexicographical sort is that the function used to compare tuples is different. This difference does not affect scalability with respect to the number of tuples even though it makes \textsc{Vortex} slower. As for 
$\textsc{Multiple Lists}^{\star}$, it relies on  external-memory sorting and the repeated application of \textsc{Multiple Lists} on 
fixed-sized 
blocks---both of which are scalable.

%# dd if=/dev/zero of=/tmp/output.img bs=8k count=256k
%262144+0 records in
%262144+0 records out
%2147483648 bytes (2.1 GB) copied, 6.28507 s, 342 MB/s
%mint ~ # rm /tmp/output.img

\subsection{Experimental results on realistic data sets}
\label{sec:runsonrealisticdatabaset}

We present the results in Table~\ref{table:owenstablereal}, 
giving
the compression ratio on top of the lexicographical order: e.g., a value
of two indicates that the compression ratio is twice as large as what we get with the lexicographical order.

\begin{table}
\tbl{\label{table:owenstablereal}Compression ratio compared to lexicographical sort }{%
%%%%%%%%%%%
% To reproduce, go to DeltaCoding/DeltaRunCppCode/results-tods-june2011 
% and type ./summarize.py
%%%%%%%%%%%%
\begin{tabular}{cc}
\subfloat[Census 1881]{%
\begin{tabular}{l|cc}
& \textsc{Vortex}	& $\textsc{Multiple Lists}^{\star}$	\\\hline
Sparse& \textbf{1.27} & 1.00 \\
Indirect&  \textbf{1.07} & \textbf{1.02} \\
Prefix&  \textbf{1.45} & 0.99 \\
LZO&  0.99 & \textbf{1.07} \\
RLE&  \textbf{1.02} & \textbf{1.13} \\
\textsc{RunCount} & 0.99 & \textbf{1.13} \\
\end{tabular}
}
&
\subfloat[Census-Income]{%
\begin{tabular}{l|cc}
& \textsc{Vortex}	& $\textsc{Multiple Lists}^{\star}$	\\\hline
 & 1.00 & 0.99 \\
&  \textbf{1.03} & 0.96 \\
& \textbf{1.15} & 0.95 \\
&  0.97 & \textbf{1.04} \\
&  0.97 & \textbf{1.24} \\
&  0.96 & \textbf{1.25} \\
\end{tabular}
}\\
\subfloat[Wikileaks]{%
\begin{tabular}{l|cc}
& \textsc{Vortex}	& $\textsc{Multiple Lists}^{\star}$	\\\hline
Sparse & \textbf{1.15} & 0.94 \\
Indirect &  0.63 & 0.96 \\
Prefix & \textbf{1.21} & 0.91 \\
LZO &  0.63 & \textbf{0.90} \\
RLE &  0.71 & \textbf{1.14} \\
\textsc{RunCount}& 0.70  & \textbf{1.14} \\
\end{tabular}
}&
\subfloat[SSB (DBGEN)]{%
\begin{tabular}{l|cc}
& \textsc{Vortex}	& $\textsc{Multiple Lists}^{\star}$	\\\hline
& \textbf{1.00} & 0.99 \\
&  \textbf{1.00} & 0.99 \\
&  \textbf{1.00} & 0.97 \\
&  \textbf{1.00} & \textbf{1.00} \\
&  \textbf{1.00} & \textbf{1.00} \\
 & \textbf{1.00} & \textbf{1.00} \\
\end{tabular}
}\\
\\ \subfloat[Weather]{%
\begin{tabular}{l|cc}
& \textsc{Vortex}	& $\textsc{Multiple Lists}^{\star}$	\\\hline
Sparse & 0.80 & \textbf{1.06} \\
Indirect &  0.78 & \textbf{1.55} \\
Prefix & 0.74 & \textbf{0.94} \\
LZO&  0.91 & \textbf{1.96} \\
RLE&  0.67 & \textbf{1.69} \\
\textsc{RunCount} &  0.66 & \textbf{1.67}\\
\end{tabular}
}&
\subfloat[USCensus2000]{%
\begin{tabular}{l|cc}
& \textsc{Vortex}	& $\textsc{Multiple Lists}^{\star}$	\\\hline
 & \textbf{1.09} & 1.00 \\
&  \textbf{1.72} & 1.06 \\
& \textbf{1.81} & 1.08 \\
&  0.26 & \textbf{0.94} \\
&  \textbf{3.08} & 1.15 \\
&  \textbf{3.04} & 1.15\\
%
%
%sorted:
%sparse 30.8564
%indirect 13.3066
%prefix 39.5686
%LZO 7.97335
%RLE 14.6159 
%Runcount 7232730
%
%vortex :
%sparse 28.2013
%indirect 7.72502
%prefix 21.9056
%LZO 30.5088
%RLE 4.74452
%runcount: 2376580
%
%
%multiple list :
%sparse 30.8689 
%indirect 12.5929
%prefix 36.4966
%LZO 8.44477 
%RLE 12.6965
%runcount: 6286755
%
%
%
\end{tabular}
}
\end{tabular}
}
\end{table}

Neither \textsc{Vortex}
nor $\textsc{Multiple Lists}^{\star}$ was able to improve
the compression ratio on the SSB data set. In fact, there was no change
(within 1\%) when replacing the lexicographical order with \textsc{Vortex}.
And $\textsc{Multiple Lists}^{\star}$ made things slightly worse (by 3\%)
for Prefix coding, but left other compression ratios unaffected (within 1\%).
To interpret this result,  consider that, while widely used, the DBGEN tool 
still generates synthetic data.  For example, out of 17~columns, seven have
almost perfectly uniform histograms.
Yet  ``real world data is not uniformly distributed''~\cite{Poess:2003:DCO:1315451.1315531}. 

For Census1881, the most remarkable result is that \textsc{Vortex}
was able to improve the compression under Sparse or Prefix coding by 27\% and 45\%.
For Census-Income,  $\textsc{Multiple Lists}^{\star}$ was able
to improve RLE compression by 25\%.
For Wikileaks, we found it interesting  that 
$\textsc{Multiple Lists}^{\star}$  reduced \textsc{RunCount} (and the RLE output)
by 14\% whereas the lexicographical order is already 1.26-optimal, which means
that $\textsc{Multiple Lists}^{\star}$  is 1.1-optimal in this case.
On Weather, the performance of \textsc{Vortex} was disappointing: it worsened
the compression in all cases. However, $\textsc{Multiple Lists}^{\star}$
had excellent results: it doubled the Lempel-Ziv (LZO) compression, and it improved
RLE compression by 70\%.  %However, for 
Yet on the 
USCensus2000 data set,  \textsc{Vortex} was preferable to $\textsc{Multiple Lists}^{\star}$ for all compression schemes except LZO\@. We know that the lexicographical order is 3.4-optimal at reducing \textsc{RunCount}, yet \textsc{Vortex} is able to reduce \textsc{RunCount} by a factor of 3 compared to the lexicographical order. It follows that  \textsc{Vortex} is 1.1-optimal in this case.

Overall,  both \textsc{Vortex} and $\textsc{Multiple Lists}^{\star}$
can significantly improve over the lexicographic order when a 
database-compression technique is used on real data.  For every 
database-compression technique, significant improvement could be obtained by at least one
of the reordering heuristics on the real data sets.  However, significant degradation could
also be observed, and lexicographic order was best
in four realistic cases (Sparse coding on Census-Income, 
Indirect coding on  Wikileaks, Prefix  coding on Weather and LZO on USCensus2000).  
In \S~\ref{sec:guidance}, we propose to determine, based on characteristics of the data, whether significant gains are possible on top of the lexicographical order.

\subsubsection{Our row-reordering heuristics are scalable}
We present wall-clock timings in Table~\ref{table:owenstablerealspeed}
 to confirm our claims of scalability.
For this test, we included a variation of the SSB
where we used a scale factor of 100 instead of
40 when generating the data. That is, it is 2.5~times
larger (henceforth SSB $\times 2.5$). 
 As expected,
the lexicographical order is 
fastest,
whereas  $\textsc{Multiple Lists}^{\star}$
is slower than either \textsc{Vortex} or the lexicographical order. 
On the largest data set (SSB), \textsc{Vortex} and $\textsc{Multiple Lists}^{\star}$ were 
3~and
 4~times slower than lexicographical sorting.
One of the benefits of an approach based on partitions, such as 
$\textsc{Multiple Lists}^{\star}$,  is that one might stop early if benefits
are not apparent after a few partitions.
When comparing SSB and SSB $\times 2.5$, we see that the
running time grew by a factor of 4 for the lexicographical order, a factor of 2 for \textsc{Vortex} and a factor of 
 3 for $\textsc{Multiple Lists}^{\star}$.
  For SSB~$\times 2.5$, the running time of  $\textsc{Multiple Lists}^{\star}$ included 50\,min for sorting the table lexicographically, and the application of \textsc{Multiple Lists} on blocks of rows only took 104\,min. Because $\textsc{Multiple Lists}^{\star}$ 
uses blocks
with a fixed size, its running time will be eventually dominated by the time required to sort the table lexicographically as we increase the number of tuples.

\begin{table}\tbl{\label{table:owenstablerealspeed}Time necessary to reorder the rows}{%
\begin{tabular}{l|lll}%
& \textsc{Lexico.}&	\textsc{Vortex}& $\textsc{Multiple Lists}^{\star}$	\\\hline
Census1881 & $2.7$\,s & $14$\,s &  $19$\,s\\
Census-Income& $0.19$\,s& $4.7$\,s & $11$\,s   \\
Wikileaks &  $0.3$\,s &  $2.0$\,s&  $2.5$\,s \\
SSB& 12\,min & 35\,min & 52\,min \\
SSB $\times 2.5$ & 49\,min & 105\,min &  154\,min \\
Weather &  6\,min  & 26\,min & 43\,min\\
USCensus2000 & 33\,s & 3\,min & 12\,min\\
\end{tabular}}
\end{table}

\subsubsection{Better compression improves speed}
Everything else being equal, if less data needs to be loaded from RAM and disk to the CPU,  speed is improved.
It remains to assess whether improved compression can translate into better speed in practice. 
Thus,
we  
evaluated how fast we could uncompress all of the columns back into the original 32-bit dictionary values.
Our test was to retrieve the data from disk (with buffering) and store it back on disk. We report the ratio of the decompression time with lexicographical sorting over the decompression time with  alternative row reordering methods. Because the time required to write the decompressed values 
would have been 
unaffected by the compression, we would not expect speed gains exceeding 50\% with better compression in this kind of test.
\begin{itemize}
\item  First, we look at our good compression results on the Weather data set with $\textsc{Multiple Lists}^{\star}$ for the LZO and RLE (resp. 1.96 and 1.69 compression gain). The decompression speed was improved respectively by a factor of 1.19 and 1.14. 
\item  Second, we consider the USCensus2000 table, where \textsc{Vortex} improved both Prefix coding and RLE compression (resp. 1.81 and 3.04 compression gain). We saw gains to the decompression speed of 1.04 and 1.12.
\end{itemize}
These speed gains were on top of the gains already achieved by lexicographical sorting. For example, Prefix coding was only improved by 4\% compared to the lexicographical order on the USCensus2000 table, but if we compute ratios with respect to a shuffled table, they went from 1.25 for  lexicographical sorting
to 1.30 with \textsc{Vortex}. Hence, the total performance gain due to row reordering is 30\%.

\subsection{Guidance on selecting the row-reordering heuristic}
\label{sec:guidance}

It is difficult to determine
which row-reordering heuristic is best given a table and a compression scheme. 
Our processing
techniques are already fast, and useful guidance would need to be obtained faster---probably limiting us to decisions
using summaries such as those  maintained by the DBMS\@. And such concise summaries might be insufficient:
\begin{mylongitem}
\item Suppose that we are given a set of columns and complete knowledge of their histograms. That is, we have the list of attribute values and their frequencies. Unfortunately, even given all this data, we could not predict the efficiency of the row reordering techniques reliably. Indeed, consider the USCensus2000 data set. 
According to Table~\ref{table:owenstablereal}, \textsc{Vortex} improves RLE compression by a factor of 3 over the lexicographical order. Consider what happened when we took the same table (USCensus2000) and randomly shuffled columns, independently. The column histograms were not changed---only the relationships between columns were affected. Yet, not only did \textsc{Vortex} fail to improve RLE compression over this newly generated table, it made it much worse (from a ratio of 3.04 to 0.74). The performance of $\textsc{Multiple Lists}^{\star}$ was also adversely affected: while it  slightly improves the compression by Prefix coding (1.08) over the original USCensus2000 table, it made compression worse (0.93) over the reshuffled USCensus2000 table. 
\item Perhaps one might hope to predict the efficiency of row-reordering techniques by using small samples, without ever sorting the entire table. There are reasons again to be pessimistic.  We took a random sample of 65\,536~tuples from the USCensus2000 table. Over such a sample, \textsc{Vortex} improved LZO compression by 2.5\% compared to the lexicographical order, whereas over the whole data set \textsc{Vortex} makes LZO much worse than the lexicographical order (1.025 versus 0.26). Similarly, whereas \textsc{Vortex} improves RLE by  a factor of 3 when applied over the whole table, the gain was far more modest over our sample (1.06 versus 3.04).
\end{mylongitem}

However, we can offer some guidance. For compression schemes that 
are closely related to \textsc{RunCount}, such as RLE, the optimality of a lexicographic
sort should be computed using Lemma~\ref{lemma:omegabound}. If $\omega \approx 1$, we can
safely conclude that the lexicographical order is sufficient.

Moreover, our results on synthetic data sets (\S~\ref{sec:experiments}) suggest that some statistical dispersion in the frequencies of the values is necessary. Indeed, we could not improve the \textsc{RunCount} of tables having uniformly distributed columns  
even when $\omega$ were relatively large.
On our real data sets, 
we got the best compression gains compared to the lexicographical order with the Weather and USCensus2000 tables. They both have high $p_0$ values (0.36 and 0.54). 

Hence, we propose to  only try better row-reordering heuristics when $\omega$ and $p_0$ are large (e.g., $\omega>3$ and $p_0>0.3$). Both measures can be computed efficiently. 

 Furthermore, when applying a scheme
such as $\textsc{Multiple Lists}$ on partitions of the sorted table, it would be reasonable to
stop the heuristic after a few partitions if there is no benefit. For example, consider the Weather data and  $\textsc{Multiple Lists}^{\star}$. After 20~blocks of 131\,072~tuples, we have a promising gain of 1.6 for LZO and RLE, but a disappointing ratio of 0.96 for Prefix coding. That is, we have valid estimates (within 10\%) of the actual gain over the whole data set after processing only 2\% of the table.

%	5.23441 603 145
%SparseCoding<128>	5.67076 882 164
%IndirectCoding<128>	5.24319 2076 61
%PrefixCoding<128>	4.64206 553 44
%lzo(v.2.05)	6.73583 173 299
%NaiveDictCODEC	15 975 95
%OptimizedRunLengthEncoding	3.23136 663 64
%NaiveRunLengthEncoding	11.6174 638 102
%# got RunCount = 6720412
%############################
%# multiple list  
%# multiplelists sorting with blocks of size 131072
%# 
%# computing DictCODEC ... 
%# computing SparseCoding<128> ... 
%# computing IndirectCoding<128> ... 
%# computing PrefixCoding<128> ... 
%# computing lzo(v.2.05) ... 
%# computing NaiveDictCODEC ... 
%# computing OptimizedRunLengthEncoding ... 
%# computing NaiveRunLengthEncoding ... 
%DictCODEC	5.23441 591 85
%SparseCoding<128>	5.27688 830 110
%IndirectCoding<128>	3.41911 1627 47
%PrefixCoding<128>	4.83552 549 27
%lzo(v.2.05)	4.25617 68 225
%NaiveDictCODEC	15 951 38
%OptimizedRunLengthEncoding	2.01255 616 31
%NaiveRunLengthEncoding	7.4274 605 66
%# got RunCount = 4327669

\section{Conclusion}

For the TSP under the Hamming distance, lexicographical sort is an
effective and natural heuristic. It appears to be easier to
surpass the lexicographical sort when the column histograms have
high statistical dispersion (e.g. Zipfian distributions).

Our original question was whether engineers willing to spend extra time reordering rows could improve
the compressibility of their tables, at least by a modest amount. 
Our answer is positive.
\begin{itemize}
\item Over real data, $\textsc{Multiple Lists}^{\star}$ always improved
RLE compression when compared to the lexicographical order  (10\% to 70\% better). 
\item \textsc{Vortex} almost always improved Prefix coding compression, sometimes
by a large percentage (80\%) compared to the lexicographical order.
\item On one data set, \textsc{Vortex} improved RLE compression by a factor of 3 compared to lexicographical order.
\end{itemize}  

As far as heuristics are concerned, we have certainly not exhausted the possibilities.
Several tour-improvement heuristics used to solve the TSP~\cite{johnsonmcgeoch1997}  could
be adapted for row reordering. 
Maybe more importantly, we could adapt the TSP heuristics using a different distance measure
than the Hamming distance. For example, consider difference coding~\cite{moffat2000binary,db2luw2009,Anh:2010:ICU:1712666.1712668} where the successive differences
between attribute values are coded. In this case, we could use 
an inter-row distance that
measures the number of bits required to code the differences. 
Just as importantly, the implementations of our row-reordering
heuristics are sequential: parallel versions could be faster, especially
on multicore processors.
\bibliographystyle{acmsmall} 

%eventually, insert the bibtex-generated file here, to make this file self contained

%ACM seems to want short titles for journals, but long titles for conferences
% with "8th" rather than "Eighth" etc.

%\bibliography{../bib/longtitles,../bib/shorttitlesiso,../bib/acmtitles,../bib/overwritewithacm,../bib/lemur}

\received{M Y}{M Y}{M Y} %per acmsmall

\end{document}